\documentclass[lettersize,journal]{IEEEtran}

\usepackage{hyperref}
\usepackage[all]{hypcap} 
\usepackage{color,overpic,cite}
\usepackage{xcolor}
\usepackage{pdfsync}
\usepackage{url}
\usepackage{amsmath,amsfonts,amssymb,amsthm,bbm,bm,comment,multirow, subfigure}
\usepackage{tcolorbox}
\usepackage{siunitx}
\usepackage{mdframed}
\usepackage{bm}
\usepackage[ruled,commentsnumbered]{algorithm2e}
\usepackage{tabularx,booktabs}
\usepackage[normalem]{ulem}
\usepackage{setspace}
\usepackage{enumitem}
\usepackage{amssymb} 
\usepackage{pifont}
\usepackage{algorithmic}
\usepackage{dsfont}
\usepackage{float} 

\usepackage{subcaption}

\newtheorem{lemma}{Lemma}

\renewcommand{\baselinestretch}{0.98} 

\def\BibTeX{{\rm B\kern-.05em{\sc i\kern-.025em b}\kern-.08em
    T\kern-.1667em\lower.7ex\hbox{E}\kern-.125emX}}

\newcommand{\rev}[1]{{\color{black}#1}} 
\newcommand{\minorrev}[1]{{\color{black}#1}} 

\begin{document}

\title{Optimizing Conditional Handovers via Meta-Learning}
\title{Learning-Driven Handover Management for 6G: Joint Optimization of Traditional and Conditional HOs in O-RAN}
\title{Meta-Learning-Enhanced Mobility Management in O-RAN}
\title{Meta-Learning for Handover Optimization in O-RAN: A Unified Framework for NextGen Mobility}
\title{Meta-Learning for Handover Control in NextGen O-RAN}
\title{Meta-Learning-Based Handover Management in NextG O-RAN}

\author{Michail~Kalntis, ~George~Iosifidis, ~Jos\'e~Su\'arez-Varela, ~Andra~Lutu, ~Fernando~A.~Kuipers
\thanks{M. Kalntis, G. Iosifidis, and F. A. Kuipers are with Delft University of Technology, The Netherlands (\{m.kalntis, g.iosifidis, f.a.kuipers\}@tudelft.nl).

J. Su\'arez-Varela and A. Lutu are with Telef\'onica Research, Spain (emails: \{jose.suarez-varela, andra.lutu\}@telefonica.com).}}


\maketitle

\begin{abstract}
While traditional handovers (THOs) have served as a backbone for mobile connectivity, they increasingly suffer from failures and delays, especially in dense deployments and high-frequency bands. To address these limitations, 3GPP introduced Conditional Handovers (CHOs) \rev{that enable} proactive cell reservations and user-driven execution. However, both \rev{handover (HO)} types present intricate trade-offs in signaling, resource usage, and reliability. This paper presents \rev{unique}, countrywide mobility management datasets from a \rev{top-tier} mobile network operator \rev{(MNO)} that offer fresh insights into these issues and call for adaptive and robust HO control in next-generation networks. Motivated by these findings, we propose \texttt{CONTRA}, a \rev{framework that, for the first time, jointly optimizes THOs and CHOs within the \mbox{O-RAN} architecture}. \rev{We study two variants of \texttt{CONTRA}: one where users are a priori assigned to one of the HO types, reflecting distinct service or user-specific requirements, as well as a more dynamic formulation where the controller decides on-the-fly the HO type, based on system conditions and needs.} To this end, it relies on a practical meta-learning algorithm that adapts to runtime observations and guarantees performance comparable to an oracle with perfect future information (universal no-regret). \texttt{CONTRA} is specifically designed for near-real-time deployment as an \mbox{O-RAN} xApp and aligns with the 6G goals of flexible and intelligent control. Extensive evaluations leveraging \rev{crowdsourced} datasets show that \texttt{CONTRA} improves user throughput and reduces \rev{both THO and CHO switching costs}, outperforming 3GPP-compliant \rev{and Reinforcement Learning (RL)} baselines 
in dynamic and real-world scenarios.
\end{abstract}


\IEEEpeerreviewmaketitle

\section{Introduction}\label{sec:intro}

\subsection{Motivation}

Mobility management is a cornerstone of mobile communication systems, ensuring service continuity as users move through diverse coverage areas \rev{of deployed \textit{cells}}. At the heart of this capability lies the \emph{handover} (HO) mechanism, which transfers active sessions from one cell to another. Despite decades of optimization and standardization \cite{3gpp_38_300, 3gpp_36_300}, HOs continue to face significant challenges in today's networks. In particular, dense deployments, heterogeneous radio access technologies (RATs), and the shift to higher frequencies increase the likelihood of HO delays and failures \cite{kalntis_imc24, stanczak22}. These issues are compounded by the increasingly dynamic nature of user mobility and volatile signal conditions, which are suboptimally handled by existing \emph{reactive} traditional HO (THO) procedures, as evidenced in our analysis with real-world data in Sec. \ref{sec:data_collection_analysis}.

While the 5G architecture introduced flexibility \cite{wanshi_5G, Wong_Schober_Ng_Wang_2017}, the fundamental approach to HO has remained largely unchanged. However, the 6G vision, expressed in various 3GPP workshops and white papers \cite{3gpp-6g, ericsson_perf_guarantees, chen_5G6G, wanshi_5G6G, 3gpp_HO_vestel},
paints a different picture, where mobility should be \emph{intelligent, proactive, and 
relying on native Artificial Intelligence (AI)}. 
To support highly demanding 6G use cases, the network must adapt in real time to user requirements and the current network conditions~\cite{ericsson_robustness}.
Specifically, seamless mobility across different service domains, ultra-reliable, low-latency HO, and resource efficiency are identified as critical enablers for 6G \cite{wanshi_5G6G, chen_5G6G}. We thus call for a \emph{paradigm shift} in mobility management, underpinned by 
(near) real-time data-driven control.

\begin{figure}[t]
    \centering
    \includegraphics[width=0.51\columnwidth]{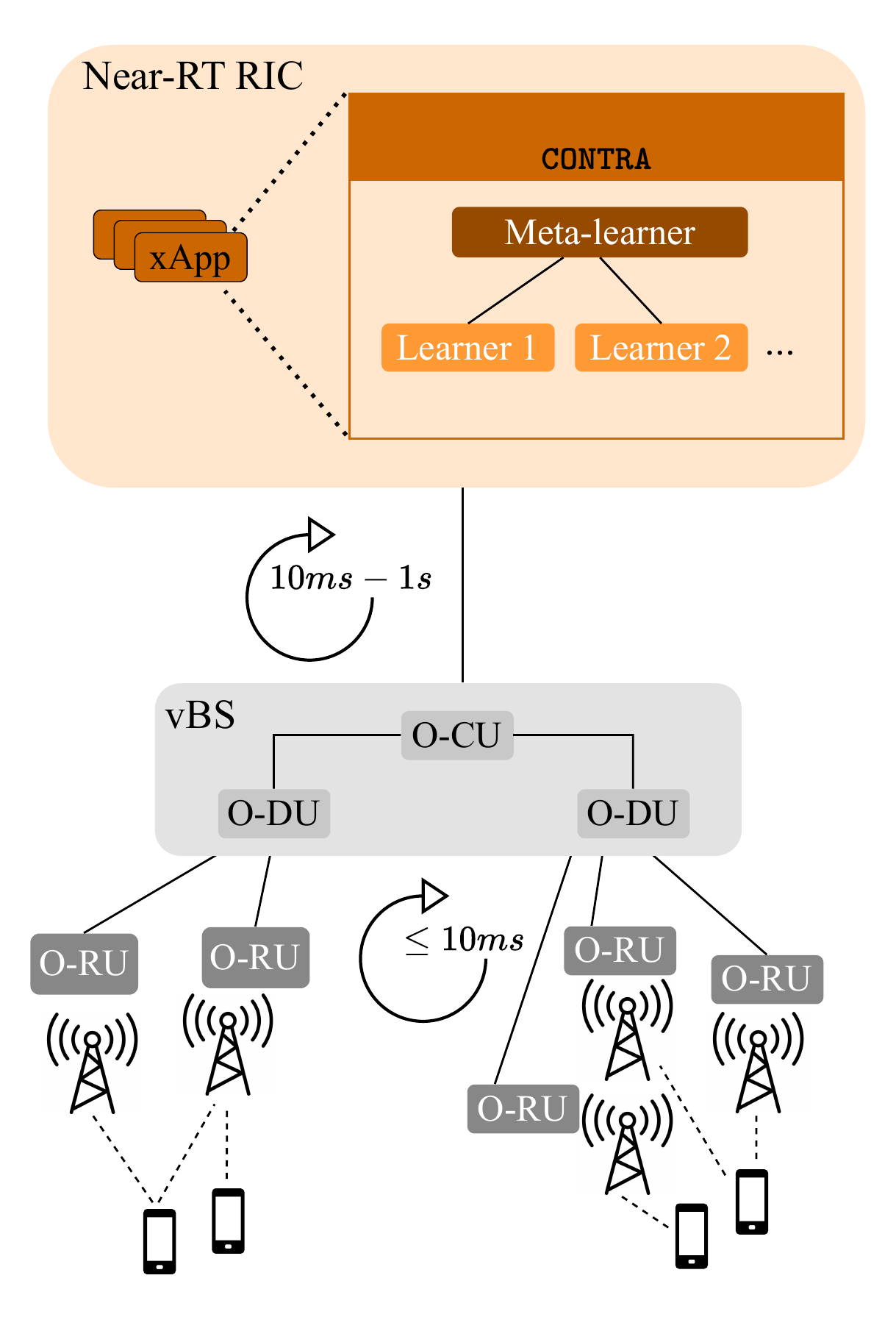}
   \vspace{-1mm}
    \caption{\texttt{CONTRA} deployed as an xApp in near-RT O-RAN RIC.}
    \vspace{-2mm}
    \label{fig:oran_architecture}
\end{figure}

To mitigate these limitations, 3GPP introduced Conditional Handovers (CHOs) \cite{3gpp_38_300, ericsson_CHO_20}. The main idea is to \textit{proactively} reserve resources
in multiple \textit{candidate} cells while signal conditions are favorable, and delegate the final HO decision to the user. This novel approach of \textit{network-configured, user-decided HOs}, instead of the traditional \textit{network-configured, user-assisted} HOs, has shown to reduce HO delays and failures \cite{iqbal_beamforming23, leo_juan22}. At the same time, CHOs introduce new challenges: selecting the optimal set of cells to prepare is highly context-sensitive; preparing too many cells leads to resource
overutilization, while preparing too few risks service interruption due to fallback to traditional HOs~\cite{prado_CHOopt23, stanczak22, iqbal22, martikainen18}. As CHO policies remain non-adaptive in practice, the trade-offs they involve are often suboptimally managed.

In this paper, we argue that next-generation (NextG) networks will benefit from AI-based solutions that jointly optimize traditional and conditional HOs and adapt the HO strategy to the current user requirements. In fact, 6G early releases are expected to include both network-triggered and user-triggered HO solutions \cite{ericsson_perf_guarantees}. At the same time, the rise of \emph{O-RAN} provides a practical vehicle for deploying such intelligence since HO logic can now be implemented as near-real-time xApps, enabling a flexible and programmable AI-native radio access network intelligent controller (near-RT RIC); see Figure \ref{fig:oran_architecture}. This controller communicates with the O-RAN-compliant central, distributed, and radio units (O-CU, O-DU, and O-RU, respectively), delegating the decisions to multiple cells and users. This architectural openness and AI-based control creates a timely opportunity to rethink mobility management for NextG networks. 

In summary, in this paper, we answer the question: \emph{How can we design a robust and principled AI-driven framework for jointly optimizing traditional and conditional handovers in NextG \mbox{O-RAN} architectures?} \rev{This question is motivated by the new mobility needs in 6G, technological developments in HO mechanisms, new AI-based algorithms that can potentially enable HO optimization under realistic conditions, and the lack of unified THO and CHO mechanisms.}

\subsection{Methods \& Contributions}

To answer this question, we first analyze new countrywide datasets of mobility events from a top-tier mobile network operator (MNO) in a European country. Unlike prior studies that focus on limited end-user measurement campaigns or confined regions (see Sec.~\ref{sec:related_work}), our dataset spans 13.5K cells and 40M users, and provides comprehensive visibility into {handover failures (HOFs), HO delays, signal fluctuations, and cell heterogeneity}. Employing a statistical analysis, we quantify the conditional effect on HOFs and HO delays as a function of cell's frequency band, vendor, and location/time. We observe a Pareto-like distribution in signal fluctuation: for 20\% of movements, the Signal-to-Interference-plus-Noise Ratio (SINR) varies by 100\% within 1 s, which makes them ideal candidates for CHOs. These insights confirm the need for \rev{adaptive and robust} HO control at user-level granularity and motivate the solution of this work.

\rev{We next introduce a first-of-its-kind} \emph{unified modeling framework} that jointly captures both traditional (explicit) and conditional (implicit) HO processes. Our model represents the decisions of a controller \cite{3gpp_HO_vestel} in an O-RAN setting and accounts for user throughput, cell load, signaling overhead, and service delay due to HOs (switching costs). \rev{We study two variants of \texttt{CONTRA}: one with a priori HO type assignment per user, reflecting distinct service or user-specific requirements, and another where the controller determines on-the-fly the HO type based on system dynamics.} The key novelty lies in {treating HO type selection, candidate cell preparation, and cell association as a single online learning problem without relying on knowledge from underlying stochastic processes (e.g., SINR fluctuations, cost of signaling). We model the resource trade-offs and performance dependencies between HO types and across users, enabling per-user decision-making at the granularity of O-RAN near-real-time control loops.

To solve this problem, we propose \texttt{CONTRA}, a meta-learning-based algorithm (see Figure \ref{fig:oran_architecture}) that maintains a pool of learners \rev{(i.e., experts)} tuned for different signal conditions and mobility patterns, and employs a meta-learner to track their performance. \rev{This design enables \emph{adaptation} to changing network environments, types of user equipment (UE) and cells, without prior statistics}, and ensures performance commensurate with that of an oracle with knowledge of the future. Our approach is based on the theory of online convex optimization \rev{(OCO)} \cite{hazan-book} and our formulation builds on recent advances in \emph{smoothed online learning} \rev{(SOL)} \cite{zhang-smoothed-ol}. The proposed algorithm provides theoretical guarantees in terms of \emph{expected dynamic regret}, a rigorous measure of how well the algorithm performs compared to the oracle, under a vast range of scenarios for the changing parameters. Finally, we evaluate \texttt{CONTRA} \rev{using the datasets from the top-tier MNO and actual users}. Benchmarking it \rev{against 3GPP-compliant and Reinforcement Learning (RL) based THO and CHO mechanisms}, we find that it provides higher user throughput, lower signaling cost, and improved robustness, particularly in volatile \rev{and real SINR} environments. In summary, our contributions are:

\noindent$\bullet$ \rev{We present and analyze countrywide mobility datasets to shed light on parameters, such as frequency bands, vendors, and location/time, that affect HO failures and delays.}

\noindent$\bullet$ Motivated by the 6G vision and the analyzed datasets, we model HO control as a learning problem and propose \texttt{CONTRA}, the first unified THO/CHO orchestration framework for real-time and robust HO control in NextG O-RAN. \rev{We study two variants of \texttt{CONTRA}: a static one with predefined HO types per user, and a dynamic one, where the controller decides the HO type on-the-fly.} Our solution relies on a meta-learner that is oblivious \rev{to UE types, mobility, and network conditions}, offering performance guarantees (expected dynamic regret). 

\minorrev{\noindent$\bullet$ We release the source code\footnote{\url{https://github.com/MikeKalnt/conditional_traditional_handover_management}} of our implementation online, under a permissive free software license, along with detailed documentation.}

\noindent$\bullet$ We evaluate \texttt{CONTRA} using crowdsourced data and multiple \rev{scenarios against 3GPP-compliant and RL benchmarks}. 
The experiments highlight \texttt{CONTRA}’s \emph{efficacy, deployability, and alignment with 6G goals} of intelligent \rev{Radio Access Network (}RAN\rev{)}} functions.

\noindent\textbf{Organization}. \minorrev{Sec. \ref{sec:related_work} discusses the related work} and Sec. \ref{sec:background} provides background on \rev{T}HOs and CHOs. 
\rev{Sec. \ref{sec:data_collection_analysis} presents our datasets and the findings from their analysis.} Sec. \ref{sec:system_model} introduces the main system model and problem formulation, and in Sec. \ref{sec:algorithm}, we present the proposed solution for this first problem. \rev{Sec. \ref{sec:model_extension} extends our study to the model with dynamic handover type selection} and we evaluate our solutions in Sec. \ref{sec:evaluation}. Our study concludes in Sec. \ref{sec:conclusions} and \rev{the Appendix includes lemma proofs omitted from the main text.}

\smallskip
\section{Related Work}\label{sec:related_work}

\noindent\textbf{{Measurements \& Traditional Handovers.}}
HOs are mainly studied with traces from UEs, which are inevitably limited to certain manufacturers, areas and devices \cite{vivisecting_2022}. Recent large-scale, network-side studies \cite{kalntis_imc24} provide more visibility, but do not propose models or solutions. Prior works have addressed the joint optimization of throughput and HOs \cite{andrews-association, andrews-globecom21, choi-TWC15, kelleler-jsac23} \rev{\cite{ppo_scc}}, offering important insights into association strategies under delay constraints. Our objective aligns with these studies, and we extend their models to capture the delay characteristics and costs associated with THOs and CHOs, \rev{and heterogeneity across cells and UEs}, by analyzing countrywide datasets. Separately, recent work has addressed HOs under minimal assumptions via online learning \cite{kalntis_infocom25}, but without considering CHOs. \rev{Our approach builds on these works by incorporating, for the first time, THO and CHO characteristics into a unified, adaptive learning framework.}

\noindent\textbf{{Conditional Handovers.}}
CHOs have been proposed as a solution for non-terrestrial networks \cite{leo_juan22, leo_saglam23, leo_yang24}, 5G NR-unlicensed systems \cite{stanczak22}, fast-moving users \cite{stanczak22}, as well as in beamforming and contention-free random access \cite{iqbal_beamforming23, stanczak23_cfra}. Proposals for improving the CHO mechanism include \cite{martikainen18, samsung_cho24} which tweak the CHO thresholds to decrease HOFs; \cite{fiandrino23, stanczak23_history} which use historical handover data to decide the cell preparations; and \cite{prado_echo21} which employs UE trajectory prediction to optimize the CHO decisions. Machine Learning-assisted approaches include \cite{lee_DL20, park24} which predict ideal association strategies based on SINR data or predict the SINRs values. These important works, however, propose heuristic solutions (no optimality guarantees) and/or rely on historical data for offline training of models, which in practice may be unavailable or non-representative of encountered conditions.

In contrast, motivated by our measurements that show this problem is dynamic across UEs, cells, and time, we leverage an adaptive optimization framework. Unlike other works that explicitly avoid decision changes in cell preparations \cite{iqbal22, iqbal23_FCHO_hand_blockage} or propose static optimization formulations \cite{prado_CHOopt23} which, unavoidably, rely on heavy assumptions (static and known parameters), we propose a tunable approach to the network's preference switching model that can cope with time-varying and unknown cost values. In \cite{wiopt25_kalntis}, a meta-learning approach with minimal assumptions is used but focuses only on CHOs and omits a fair scheduler for allocation, as shown in eq. \eqref{eq:utility}.

\noindent\textbf{{Meta-Learning.}}
Meta-learning is finding increasing applications in communication systems due to its robustness to distribution shifts and fast-adaptation \cite{chen2023learning}. We refer to \cite{meta-learning-infocom23} for merging proposed actions for management and orchestration (MANO) operations;  \cite{metalearning_embedding} for beamforming adaptation and MIMO systems; \cite{metalearning_qing} for user-level traffic prediction over a short time horizon; \cite{walid_uav_metalearning} for unmanned aerial vehicle (UAV) networks; \cite{metalearning_federated} for IoT devices learning together; \cite{metalearning_loadbalancing} for load balancing; and \cite{meta-learning-rl-TVT24} for handovers in vehicle-to-network communications. Here, we utilize the dynamic version of this tool, combined with online learners, so that optimal CHO and THO decisions can be learned on-the-fly, achieving rigorous theoretical and practical performance guarantees.

\section{Background}\label{sec:background}

\begin{figure}[t]
    \centering
    \includegraphics[width=0.9\columnwidth]{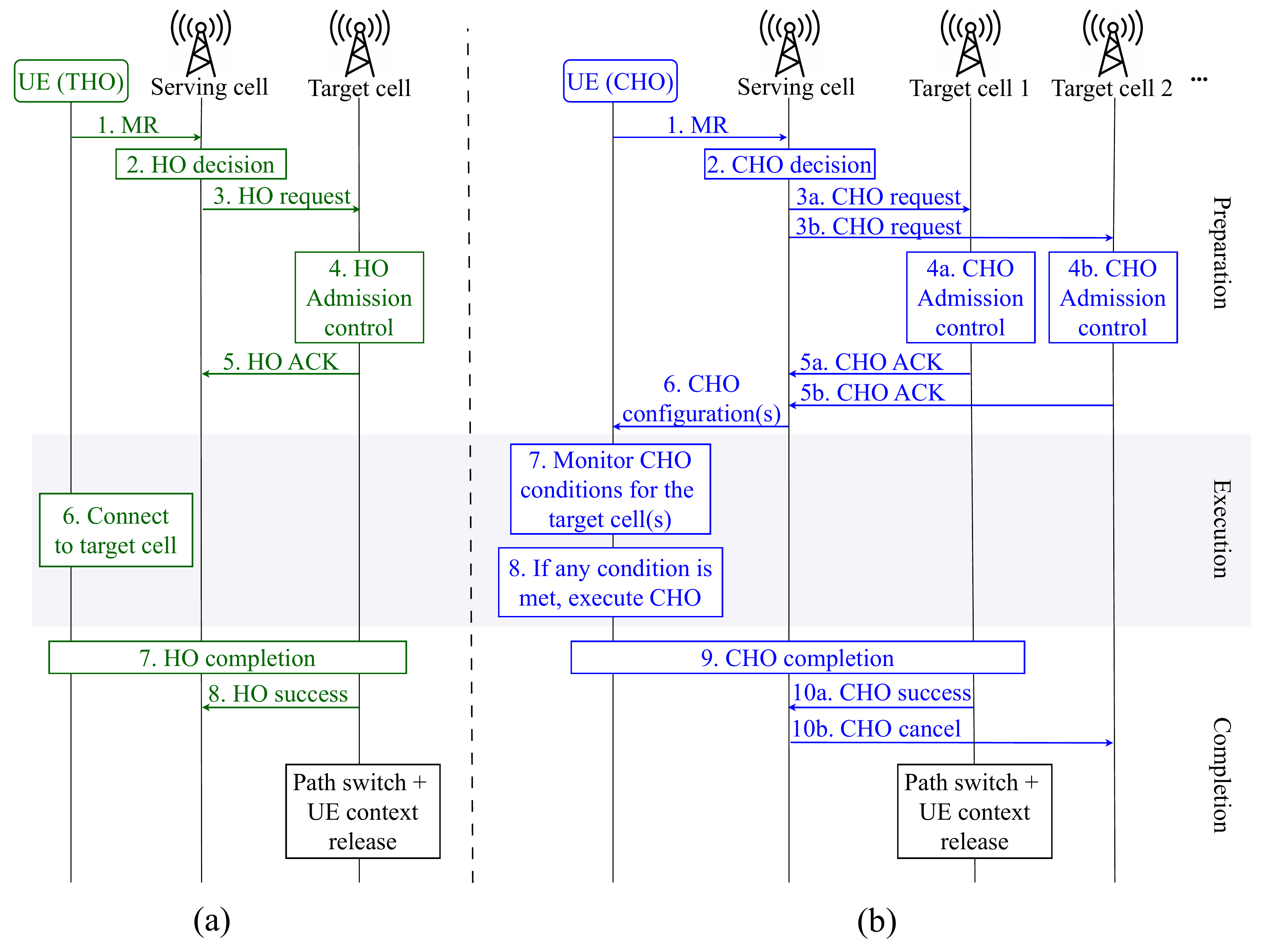}
    \vspace{-3mm}
    \caption{Basic steps/procedures of UEs performing (a) traditional (green) and (b) conditional (blue) HOs, requiring 8 and 10 main steps, respectively.}
    \vspace{-6mm}\label{fig:CHO_basic_steps}
\end{figure}

\minorrev{\noindent\textbf{Traditional HO Basics.} THOs are triggered when the signal quality of the serving cell degrades (e.g., A2 event) and/or another neighboring cell's signal becomes higher (e.g., A3 event) \cite{3gpp_36_331, 3gpp_38_331}. Specifically, the serving cell configures the UE monitoring conditions, such as hysteresis, offset, and time-to-trigger (TTT) parameters, and the UE sends a measurement report (MR) according to these (step 1), as illustrated in Figure \ref{fig:CHO_basic_steps}(a). The \textit{preparation phase} continues with the HO decision and admission (steps 2, 3, 4, 5), which dictate the \textit{one} target cell that the UE should connect to (step 6). This transition (i.e., \textit{execution phase}) is executed \textit{immediately} after the UE receives the command and finishes in the \textit{completion phase}. 

This \textit{reactive} approach, where THOs are initiated \textit{after} signal conditions deteriorate, often leads to increased HO delay and HOF rates, particularly due to the small cell density and rapid signal fading of high frequency bands ($>$ 24 GHz, FR2) \cite{3gpp_consecutive_CHO, kalntis_imc24}. As illustrated in Figure \ref{fig:hof_theory}, it is common for HOFs to occur when a user attempts to send an MR under degrading signal conditions; or even if the MR is successfully sent, the worsening signal conditions may prevent the user from receiving the HO command from the serving cell.

\minorrev{\noindent\textbf{Conditional HO Basics.}} CHOs, designed first as part of 3GPP Rel. 16, address these limitations by \textit{(i)} separating the preparation and execution phases, with their gap potentially being as large as 9--10 s \cite{iqbal22}, and \textit{(ii)} offloading part of the HO decision-making to the user before signal conditions deteriorate \cite{3gpp_38_300}. As can be seen in Figure \ref{fig:CHO_basic_steps}(b), a CHO decision is taken while signal quality is still adequate (step 2), and the source cell can pre-configure (i.e., prepare) multiple \textit{candidate} target cells (steps 3, 4, 5) based on the MR of the UE (step 1). To conclude the preparation phase, the source cell provides the monitoring conditions (step 6), and the execution phase starts: the user applies these conditions continuously (step 7) to evaluate the signal of the source and candidate target cells. 

If any of the predefined conditions are met (step 8), the UE executes the stored HO command (as if it had just been received) without an MR or reply from the serving cell. If more than one cell meets the execution condition, the UE decides which to access; typically, the one with the highest SINR. CHO is finalized in the completion phase, where resources are released from the serving cell and a new path to the target cell is established (steps 9, 10, and beyond). Once resources are reserved for a UE through admission control (step 4), they are not released until a cancellation is sent (step 10b).}

\minorrev{\noindent\textbf{Conditional HO Key Trade-Offs.}}
Determining which cells to prepare is critical in CHOs for balancing resource efficiency and mobility robustness. Ideally, the MNO, whose goal is to economize cell resources, would allow the preparation of the single cell to which the UE will connect in the next slot (i.e., the ``correct'' target cell) and whose signal strength is high, hence maximizing the UE's throughput. However, this cannot always be predicted and multiple candidate cells often need to be prepared, e.g., because the user is located at the edge of a cell and its trajectory is unknown. 
At the same time, a long list of prepared cells does not \textit{necessarily} increase the likelihood of including the ``correct'' target cell, especially if all prepared cells exhibit low signal quality. Conversely, while preparing fewer cells may save resources, it increases the risk of falling back on traditional HOs; thus, larger HO delays and HOF probability (see Sec. \ref{sec:data_collection_analysis}).

Another aspect to consider when optimizing CHOs is the \textit{signaling cost} of preparing cells. For instance, in environments such as FR2, the small cell density and rapid signal fluctuations already result in more frequent HOs, leading to significant signaling overhead. The additional burden of continuously preparing and releasing cells due to constant signal variations further exacerbates this overhead \cite{iqbal22, iqbal23_FCHO_hand_blockage}. For that reason, although initial 3GPP releases specify that UEs should release CHO candidate cells after any successful HO completion to conserve resources~\cite{3gpp_38_300}, subsequent studies suggest that this approach is not always optimal for minimizing HOFs and signaling overhead~\cite{3gpp_consecutive_CHO}.

These trade-offs emphasize that the selection process for cell preparations should focus on finding a balance
between identifying the ``best''  (in terms of signal quality) potential target cells and keeping the number of prepared cells small (i.e., limited reserved resources).

\minorrev{\noindent\textbf{Granularity.}} An essential aspect of any intelligent mechanism is the granularity at which THO/CHO decisions are made (step 2 of Figure \ref{fig:CHO_basic_steps}). Typically, these decisions are taken in hundreds of milliseconds (e.g., 100--640 ms), driven by the TTT parameter; namely, the duration for which the signal of a neighboring cell should be constantly better than that of the serving cell \cite{kelleler-jsac23, 3gpp_38_331, 3gpp_36_331, martikainen18}. This granularity aligns with the decisions of a central controller in the near-RT of O-RAN, which could even handle multiple cells simultaneously \cite{o-ran-andres}. 

\begin{figure}[!t]
    \centering
    {\label{fig:CHO_HOF_causes}\includegraphics[width=0.8\columnwidth]{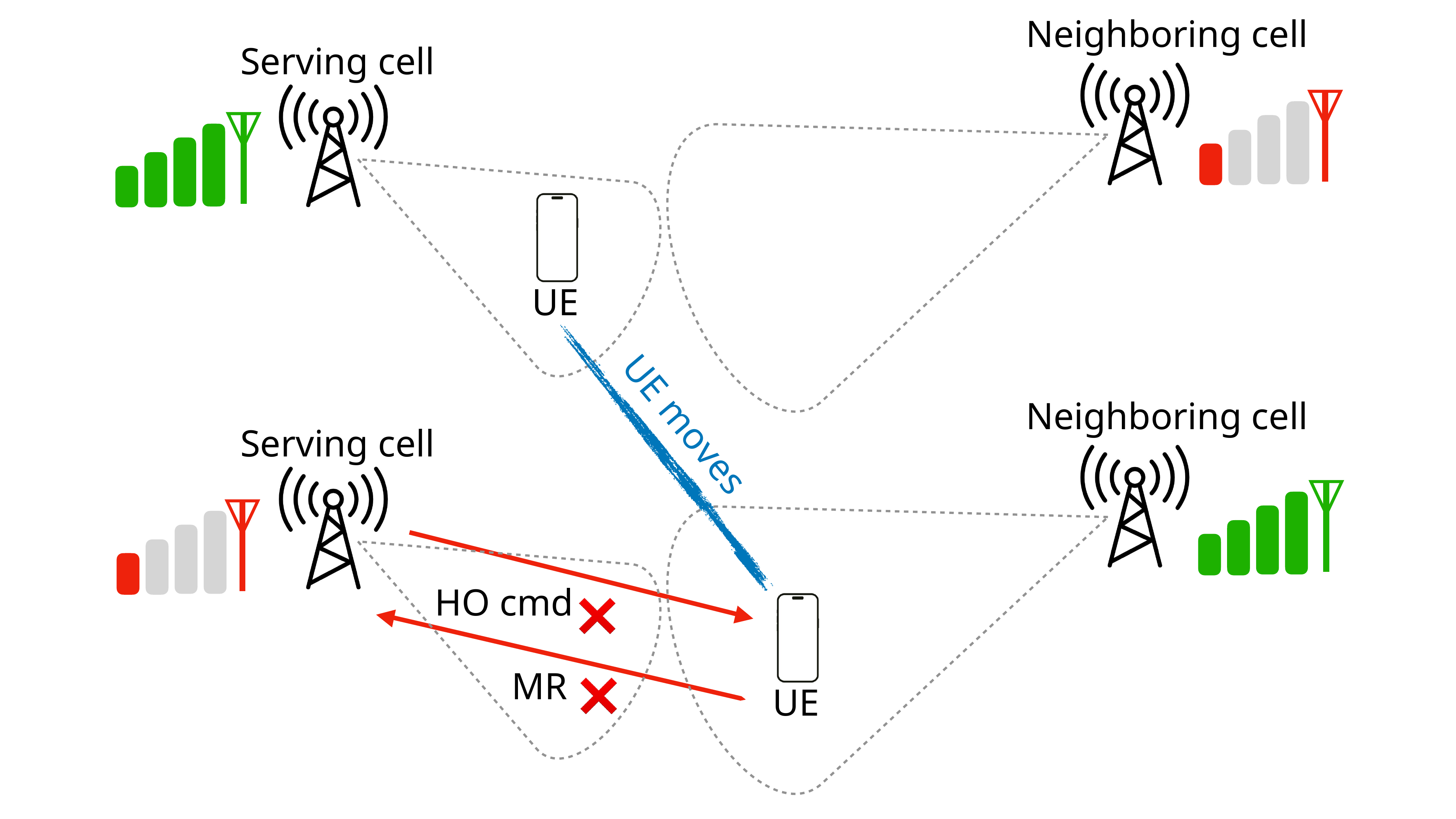}}	
    \vspace{-3mm}
    \caption{A failed traditional HO where the MR (HO) command does not reach the serving cell (UE) as the signal drops.}
    \vspace{-2mm}
    \label{fig:hof_theory}
\end{figure}

\section{Dataset \& Preliminary Analysis}\label{sec:data_collection_analysis}

To motivate the benefits of CHOs, we showcase the problems with traditional HOs using data collected from an MNO in a European country. We analyze measurements for one week, from 24-June-2024 to 30-June-2024, and focus on HO events in the most recent radio access generations: 4G and 5G. At the time this study was conducted, CHOs were unavailable in the analyzed network and the majority of 5G traffic was served through 5G-Non-Standalone (NSA), i.e., relying on the 4G Evolved Packet Core (EPC) instead of the newly deployed 5G core. Our goal is to identify the features (e.g., vendor, location) in the source and target cells, as well as in the UEs, that contribute to increased HOFs and delays. To complement our analysis, we additionally leverage: $(i)$ crowdsourced datasets of signal quality metrics per second (e.g., SINR and RSRP) of 13,590 cells (3,496 cell sites), with precise measurement locations to characterize signal fluctuations; and $(ii)$ census data to examine the correlation of population density with HOs and cell deployments.

We analyze mobility management from the perspective of \textit{both source and target} cells at the country-level.  This overcomes limitations of prior works that relied only on UE-side measurement datasets, such as \cite{vivisecting_2022, wheels_2023}, as well as other large network-side works that focus either only on measurements \cite{kalntis_imc24}, or only on the source cell characteristics, omitting the impact of target cell features apart from the change in radio access technology~\cite{kalntis_infocom25}.

\begin{figure}[t]
    \centering
    \subfigure[]{\label{fig:speed_CINR_perc_change}\includegraphics[scale=0.52]{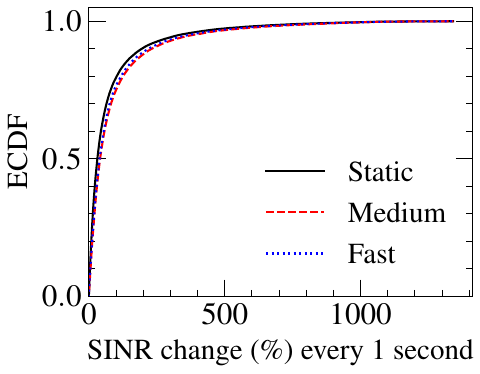}}\ \ \ 
    \subfigure[]{\label{fig:speed_RSRP_perc_change}\includegraphics[scale=0.52]{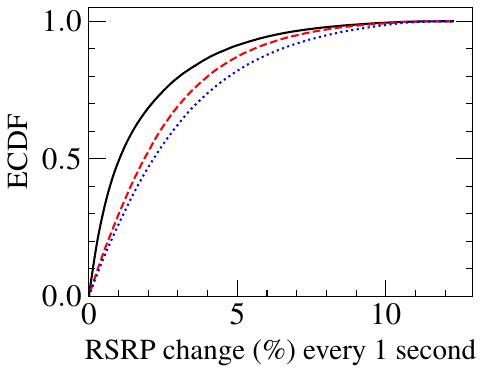}}
    \vspace{-3mm}
    \caption{(a) SINR and (b) RSRP changes per second.
    }
    \vspace{-5mm}
    \label{fig:speed_CINR_RSRP_perc_change}
\end{figure}

\subsection{Signal Quality Metrics and Mobility}

In Figure~\ref{fig:speed_CINR_RSRP_perc_change}, we examine temporal signal variations, namely SINR and RSRP, across various speeds using our crowdsourced datasets, which offer measurements at 1-second intervals. We classify samples into three mobility categories: \textit{static}, \textit{medium} speed (3--10 km/h, likely walking/running), and \textit{high} speed movements ($>$80 km/h, likely vehicular). In general, SINR is sensitive to short-term effects such as fast fading, interference, and noise floor variations, while RSRP reflects a more stable signal strength from the serving cell.

Our findings reveal that SINR exhibits considerable fluctuations, even at the second-level granularity, with more than 20\% of movements experiencing a 100\% or greater change in SINR in one second, regardless of speed. \rev{Given that our passive measurement dataset covers cell deployments in the frequency bands below 3.5 GHz, we highlight that SINR variability does not increase with mobility due to, for instance, greater Doppler spread and more dynamic interference conditions. Doppler effects typically become relevant at higher frequencies, such as mmWave and above \cite{ballesteros_doppler}, and become more noticeable in very high-speed scenarios, such as trains traveling at speeds of up to 500 km/h and Vehicle-to-Vehicle (V2V) communications \cite{ballesteros_doppler, hsp_train}. In contrast, our high-mobility tier mainly captures vehicular UEs moving at speeds above 80 km/h (and rarely exceeding 130 km/h) in an urban environment.} On the other hand, RSRP is more stable, and speed plays a bigger role; e.g., 20\% medium-speed movements have their RSRP changing more than 4\%.

It is worth recalling that in HO procedures, a typical threshold for initiating the HO based on the A3 event is commonly set around 3 dB \cite{kelleler-jsac23}. This threshold provides a useful reference when interpreting Figure \ref{fig:speed_CINR_RSRP_perc_change}, as it highlights how frequently such conditions can be met. 

\textit{\textbf{Key Takeaways}}: The SINR fluctuates by more than 100\% for over 20\% of UEs within short time intervals of 1 s. Such rapid signal variations increase the likelihood of suboptimal HO triggering, and thus, of HO delays and HOFs. The proactive cell preparation approach of CHOs can mitigate the impact of such fluctuations and improve user performance.

\subsection{Spatial Heterogeneity}
From the official census data of the analyzed country, we dissect the 300+ defined districts and 2.5M+ postcodes, and combine them with our collected cell-level datasets from the MNO to understand how cells are placed across different environments, thus affecting the number of cells that can be prepared in CHOs.
Postcodes have been classified as \textit{urban} or \textit{rural}, based on whether they are allocated to an area with a population of more or less than 10k residents, respectively.

Figure~\ref{fig:num_BS_sqkm} displays the district-level cell density (number of cells per sq. km) compared to the population density. Each point corresponds to one of the 300+ districts in the studied country and is color-coded based on the percentage of postcodes classified as urban within the district. As expected, there is a high correlation between these two factors (Pearson correlation of 0.967). In districts where more than 90\% of postcodes are urban, shown mainly in yellow, we observe 10 to 100 cells per sq. km, with areas near the capital city exceeding this range. It is interesting to note that the urban center of the capital city (highest yellow rectangle) attracts a substantial influx of non‑resident users due to being a major administrative and economic hub; consequently, operators deploy more than 650 cells per sq. km to accommodate the increased demand. On the other hand, areas with less than 20\% urban postcodes exhibit as few as 0.12 cells per sq. km, which is sufficient to meet the reduced user demand there. 

\begin{figure}[t]
    \centering
    \includegraphics[width=0.79\columnwidth]{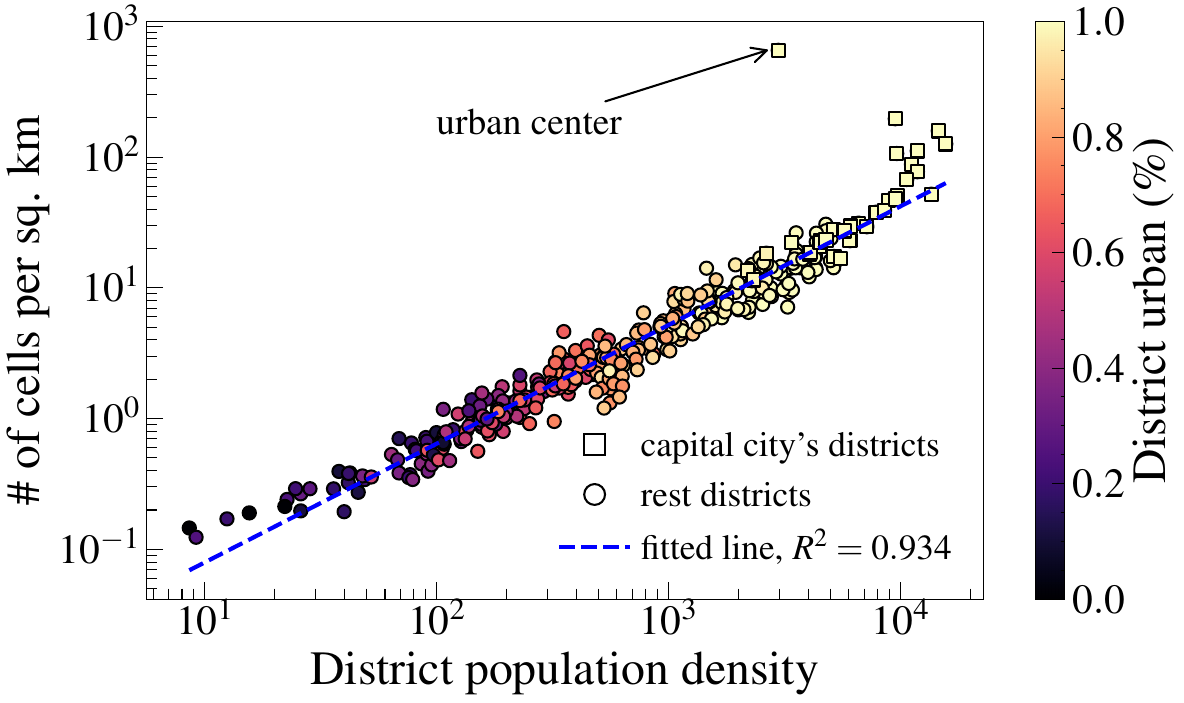}
    \vspace{-1mm}
    \caption{Number of cells per sq. km in the country (district level). Each point is colored according to the ratio of postcodes classified as urban (the others are rural).}
    \vspace{-2mm}
    \label{fig:num_BS_sqkm}
\end{figure}

In Figure~\ref{fig:heatmap_num_BS_05sqkm}, we further delve into a dense urban district to illustrate the spatial heterogeneity of the cell deployment. We partition the district into one sq. km tiles, revealing 2.2k+ total deployed cells. The per-tile cell count spans over two orders of magnitude, ranging from a single cell in sparsely covered areas to more than 160 per sq. km in highly saturated hotspots. Notably, some adjacent tiles exhibit substantial disparities, with cell counts increasing from as few as 3 cells in one tile to as many as 103 in a neighboring one. In CHOs, this means that a UE traversing these areas must dynamically account for a highly variable and often large set of candidate cells.

\textit{\textbf{Key Takeaways}}: Cellular network deployments are highly heterogeneous in terms of density, which is correlated with population. In our dataset, there are 650+ cells per sq. km in the urban center of the capital city, while in remote areas, there are as few as 0.12 per sq. km. In urban areas, the number of candidate cells for CHO can be very large and depends on cell density. There are even significant differences in cell counts between adjacent tiles (1 sq. km). This means that the signaling cost for preparing target cells is not uniform across the network, as it depends on the cell deployment density and properties (distance, frequency, etc). Therefore, the system model in Sec. \ref{sec:system_model} considers cell-dependent signaling costs.

\begin{figure}[t]
    \centering
    \includegraphics[width=0.657\columnwidth]{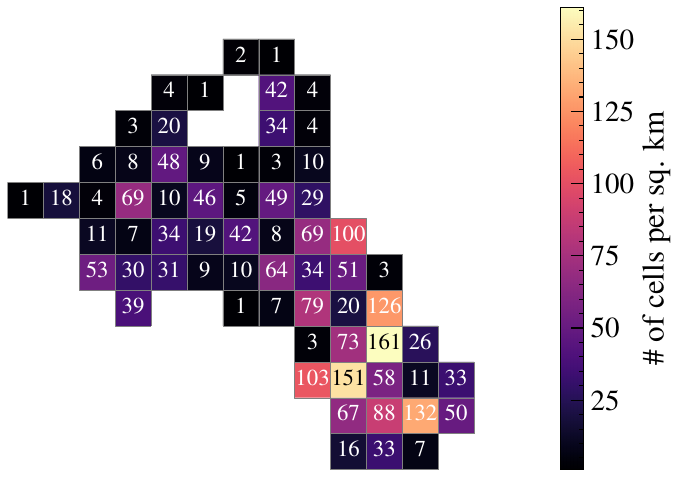}
    \vspace{-1mm}
    \caption{Number of cells/sq.km in an urban area with 2.2k cells.}
    \vspace{-2mm}
    \label{fig:heatmap_num_BS_05sqkm}
\end{figure}

\subsection{HO Failures \& Delays}
We analyze factors that increase HOFs and delays, thereby highlighting the scenarios in which CHOs offer clear benefits. 

Figure~\ref{fig:jsac_HOFrate_HOdelay_hour_boxplot} presents the hourly evolution of successful HO delays (in ms) and HOF rate (in \%) for 1 week in the entire country. Boxplots aggregate cell-level data from the same hour. The two metrics exhibit a diurnal pattern, reaching a max of 0.08\% HOF rate and 62 ms latencies (median values) at 15:00, which is the hour when most HOs occur and the network is most congested~\cite{kalntis_imc24}. Notably, from the pre-dawn median minimum (06:00)
to the max peak at 15:00, there is an increase of 252\% in the HOF rate;
at the same time, HO delay increases by 28\%.
Also, during the night hours 00:00--06:00, even though the HOF rate decreases significantly, the delay of successful HOs increases by approximately 3 ms. This is partially due to the MNO applying dynamic energy-saving policies that switch off cells acting as capacity boosters when they are not needed to meet the demand \cite{soh_onoffcells}. Clearly, existing HO mechanisms offer room for improvement, especially during peak hours, to reduce HOF rates and delays, which can be achieved by leveraging a dynamic and robust THO/CHO algorithm.

We further analyze which features of the source and target cells mainly affect the HOF rate. For that, we focus on source-target cell pairs with at least one HOF registered per day (excluding approximately 3.4\% that are outliers), and study the effect of: $(i)$~vendor, $(ii)$~frequency band, $(iii)$~model name, $(iv)$~transmit power, $(v)$~cell type (macro, micro, etc.), $(vi)$~region (north, south, etc.) and $(vii)$~area type (urban/rural). 

\begin{figure}[t]
    \centering
    \includegraphics[width=\columnwidth]{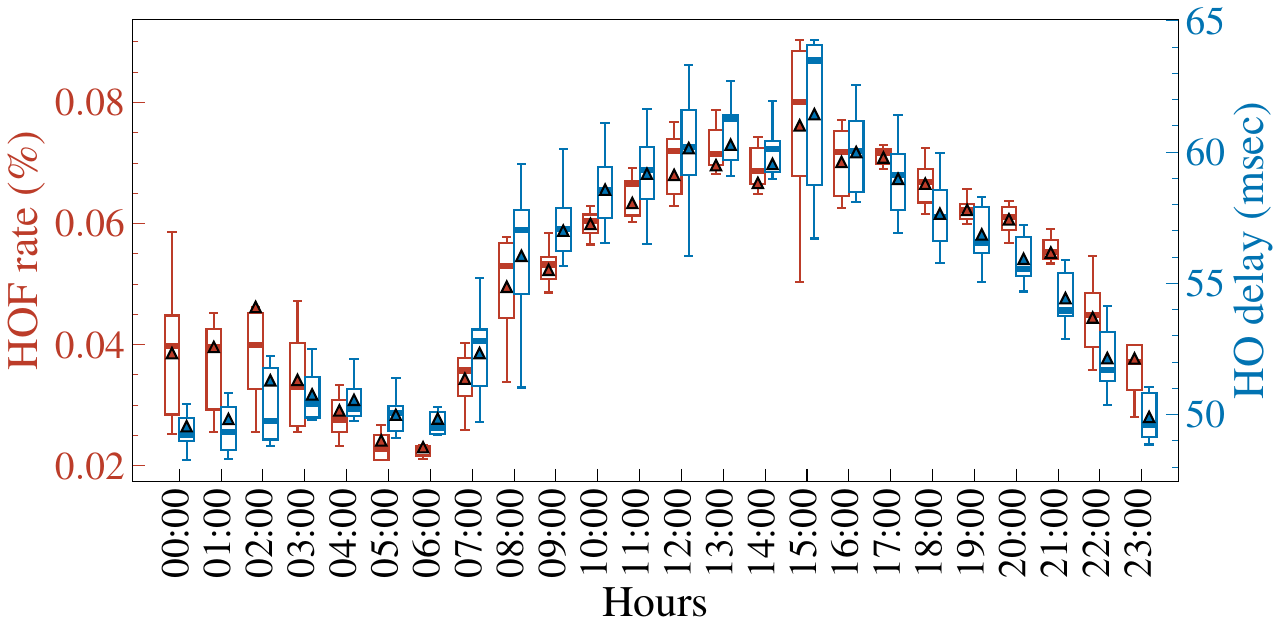}
    \vspace{-6mm}
    \caption{HOF rate (left, red y-axis) and delay of successful HOs (right, blue y-axis) per hour for a 1-week period. Triangles in boxplots depict the mean.
    }
    \vspace{-6mm}
    \label{fig:jsac_HOFrate_HOdelay_hour_boxplot}
\end{figure}

We start this analysis with a non-parametric Kruskal-Wallis test, which confirms that all examined factors affect the HOF rate ($p < 0.05$).\footnote{Normality and homoscedasticity assumptions are violated; hence, analysis of variance (ANOVA) \cite{anova} was not tested.} To determine the importance of each feature, we run multiple \rev{Machine Learning} (ML) models. We discard features with high multicollinearity (VIF $>$ 3) and Pearson correlation ($>$0.5), and include the number of HOs as an input feature to account for spatio-temporal variations (Figures \ref {fig:heatmap_num_BS_05sqkm} and \ref{fig:jsac_HOFrate_HOdelay_hour_boxplot}). The HOF rate is log-transformed, given its heavy-tailed distribution. Among the evaluated models (linear/lasso/ridge regression, k‑nearest neighbours, and random forest), the random‑forest regressor achieved the best performance, reaching an $R^2$ of 0.8 (compared to at most 0.58 for the other models) and RMSE of 0.66 (compared to at most 0.97 for the other models). As a reference, a naive model that always predicts, in the test set, the mean HOF rate from the training set, achieves an RMSE of 1.47. Permutation‑based importance reveals that the number of HOs is by far the strongest predictor, followed by the cell frequency band and the vendor; excluding these last two features drops the performance to $R^2$=0.6, confirming that frequency and vendor still contribute meaningfully.

Figure~\ref{fig:inter_intra_freq_ven} shows the effect of these latter features on HOs and HOFs. The x-axis indicates whether the HO occurred in the same (intra) or different (inter) frequency and antenna vendor, as well as if the frequencies of the source and target cells lie in the low ($\leq$ 2GHz) or mid-high ($>$ 2GHz) spectrum. The y-axis shows the HOF rate in ascending order, the delay of successful HOs and the percentage of total HO and HOF count within each category. We observe that the mean HOF rate in inter-frequency HOs, mainly when the source and/or target cells operate in low frequencies, is $\approx 0.26\%$, i.e., $\times$2.4 more than the 0.11\% of the intra-frequency ones. However, the HO delay for the former is no more than 1.5 ms higher. Intra-frequency HOs in the same low-spectrum frequency handle the majority of HOs (i.e., 42\%), but their contribution to failures closely matches this number (38.4\% of all HOFs). Lastly, note that HO delay is mainly affected by the change in vendor: even though 88\% of HOs are intra-vendor with a delay of 59--60 ms, inter-vendor HOs exhibit a delay that is 5--6 ms higher. These inter-vendor HOs are prevalent at regional borders (e.g., west-east border), given that a single vendor predominantly serves each region in the studied country. This analysis reveals that accounting for cell features is crucial for HOs.

\begin{figure}[t]
    \centering
    \includegraphics[width=\columnwidth]{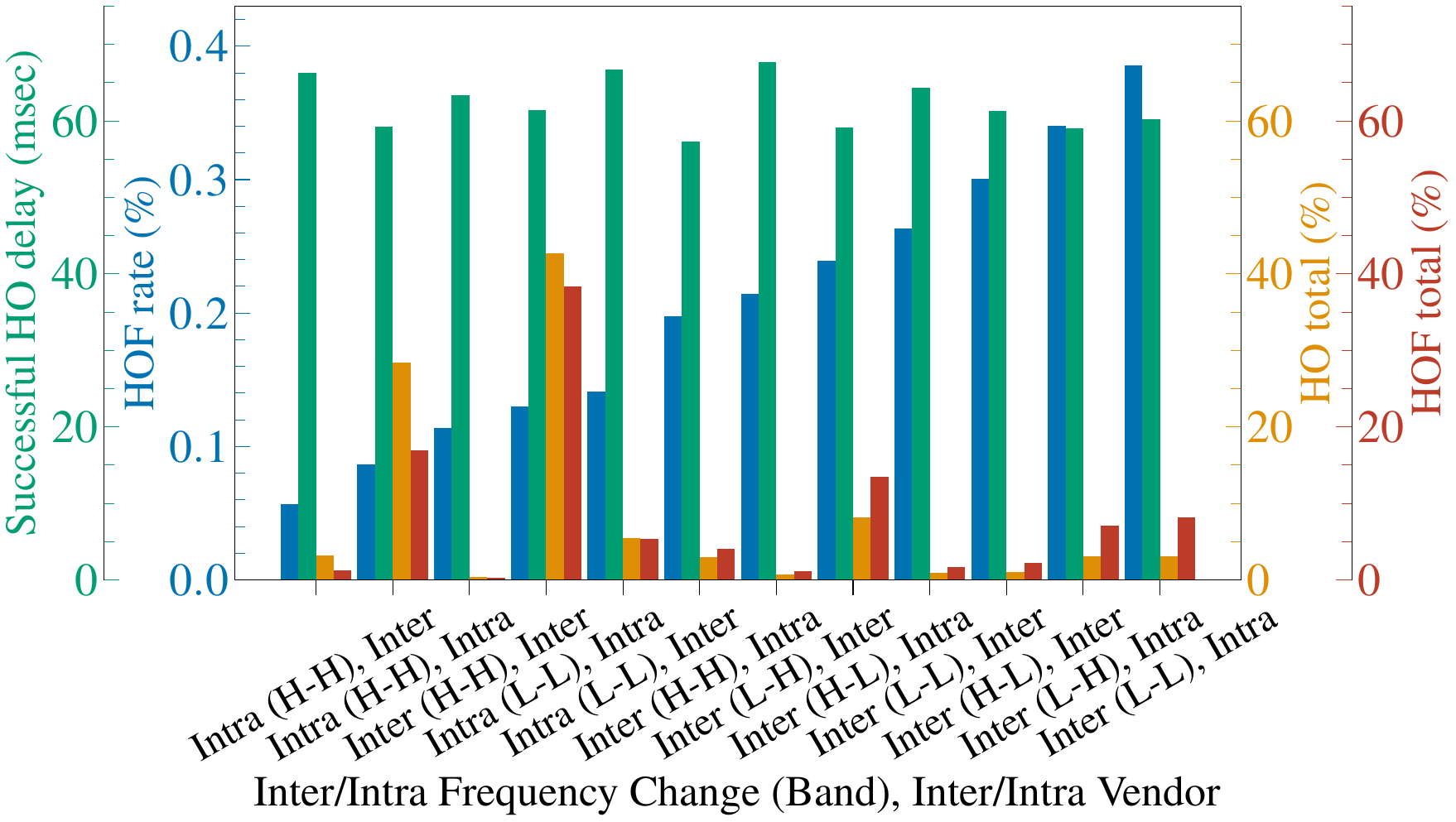}
    \vspace{-5mm}
    \caption{(left to right) Delay of successful HOs (ms), HOF rate (\%), percentage of total HOs, and HOFs based on the source-target frequency bands and vendors. 
    }
    \vspace{-6mm}
    \label{fig:inter_intra_freq_ven}
\end{figure}

\textit{\textbf{Key Takeaways}}: HO failures and delays are strongly influenced by: $(i)$~time of day, exhibiting diurnal patterns that peak at 15:00 with a 0.08\% HOF rate and a 62 ms HO delay (median values); $(ii)$ frequency changes between the source and target cells, with inter-frequency HOs having $\times$2.4 higher HOFs than intra-frequency HOs, although with a slight difference in HO delay, and $(iii)$~vendor transitions, where inter-vendor HOs exhibit 5--6 ms higher delay than intra-vendor. This evidences the potential benefit of accounting for the previous cell-level characteristics for HO optimization, as we consider in the design of \rev{our algorithm}.

\section{System Model \& Problem Formulation}\label{sec:system_model}

\subsection{Modeling Components}
We consider a heterogeneous cellular network comprising a set $\mathcal J$ of $J$ cells that serve a set $\mathcal I$ of $I$ UEs. We assume that a central controller takes decisions for multiple UEs/cells \cite{3gpp_HO_vestel} in a time-slotted manner for a set $\mathcal T$ of $T$ slots \cite{andrews-association, kalntis_tcom24} with each slot lasting hundreds of milliseconds (see Sec. \ref{sec:background}), as enabled by near-RT RIC in O-RAN\cite{o-ran-andres}.\footnote{\rev{The modeling and optimization framework of this section can be extended and implemented by other RAN systems, not solely O-RAN, as long as there is support for centralized decision-making that is fed by signal measurements.}} 

These decisions concern both THOs and CHOs, and are taken on the \textit{same time-scale}. The set $\mathcal I$ is partitioned into the set $ \mathcal I_{\text{THO}}$ of $ I_{\text{THO}}$ UEs and $\mathcal I_{\text{CHO}}$ of $ I_{\text{CHO}}$ UEs, who follow THOs (\textit{THO-enabled}) and CHOs (\textit{CHO-enabled}), respectively, where $\mathcal I_{\text{THO}}\cap \mathcal I_{\text{CHO}}=\emptyset$ and $\mathcal I_{\text{THO}} \cup \mathcal I_{\text{CHO}}\!=\!\mathcal I$. This distinction can arise due to service-level agreements (SLAs), network slicing policies \cite{slicing_Jiongyu}, or vertical-specific requirements, in which different classes of UEs are managed separately; e.g., enhanced-reliability or low-latency UEs may be assigned to CHO. Alternatively, an MNO may apply a rule-based classification, informed by prior data (e.g., historical patterns), to assign UEs to either \textit{HO type}.

For the THO-enabled UEs, we introduce the decisions:
\[\bm x_t\!=\!\big(x_{ij}{(t)} \!\in\! \{0,1\}, i\in\mathcal I_{\text{THO}},\ j\in\mathcal J\big),\]
\noindent where $x_{ij}{(t)}\!\in\!\{0,1\}$ defines the \textit{explicit one-cell association decision} with $x_{ij}{(t)}\!=\!1$ if user $i$ is assigned from the controller to cell $j$ in slot $t$. For the CHO-enabled UEs, we introduce: 
\[\bm y_t\!=\!\big(y_{ij}{(t)} \!\in\! \{0,1\}, i\in\mathcal I_{\text{CHO}},\ j\in\mathcal J\big),\] 
\noindent where $y_{ij}{(t)}\!\in\!\{0,1\}$ is the \textit{preparation decision} (\textit{implicit multi-cell association decision}): $y_{ij}{(t)}\!=\!1$ means that cell $j$ is prepared from the controller in slot $t$ for user $i$ and the user can decide whether to connect to this cell (as explained later). These decisions are drawn from the sets: 
 \begin{align*}
	&\mathcal X=\Big\{ \bm x \in \{0,1\}^{I_{\text{THO}}\cdot J}  \  \Big |   \sum_{j\in \mathcal J}{x}_{ij} = 1,  i\in \mathcal I_{\text{THO}} \Big\},
\end{align*}
\noindent where naturally each $i \in \mathcal I_{\text{THO}}$ is assigned to one cell, and:
\begin{align*}
	&\mathcal Y=\Big\{ \bm y \in \{0,1\}^{I_{\text{CHO}}\cdot J}  \  \Big |   1 \leq \sum_{j\in \mathcal J} {y}_{ij} \leq b_{i},  i\in \mathcal I_{\text{CHO}} \Big\},
\end{align*}
\noindent where no more than $b_{i}$ cells can be prepared for $i \in \mathcal I_{\text{CHO}}$.

The key metric is the SINR for the signal delivered by cell $j \in \mathcal{J}$ to any user $i \in \mathcal{I}$ in slot $t \in \mathcal{T}$:
\begin{align}
	s_{ij}{(t)}=\frac{q_j\phi_{ij}{(t)}}{	W_j\sigma^2 + \sum_{k\in \mathcal{B}_j} q_k\phi_{ik}{(t)}},
\end{align}
where $q_j$ is the transmit power of cell $j$, $\mathcal B_j$ the set of cells that operate in the same frequency as $j$, $\phi_{ij}{(t)}$ the channel gain (including pathloss, shadowing, and antenna gains), $W_j$ is the bandwidth, and $\sigma^2$ the power spectral density. In line with previous works \cite{andrews-association,kelleler-jsac23, andrews-globecom21, choi-TWC15}, $s_{ij}{(t)}$ is the average SINR in $t$, since UEs report multiple values during each slot.

We distinguish two cases for the \textit{maximum throughput} a UE can receive from a cell, based on whether it is THO- or CHO-enabled:
For the former \cite{kalntis_infocom25, andrews-association, andrews-globecom21}, this throughput is expressed as:
\begin{align}
    c_{ij}(t)=W_j\log\big(1+s_{ij}(t)\big),  i \in \mathcal I_{\text{THO}} \label{eq:rate_THO},
\end{align}
\noindent while for the latter, we introduce:
\begin{align}
    c_{ij}'(t) \!=\!
    \begin{cases}
    \! W_j \log \! \big(1 \!+\! s_{ij}(t)\big), & \!\!\!\! \text{if } j \!=\! \underset{k:\, y_{ik}(t) = 1}{\arg\max}  \ s_{ik}(t) \\
    \! 0, & \!\!\!\! \text{otherwise} \label{eq:rate_CHO}
\end{cases}
\end{align}
\noindent where $i \in \mathcal I_{\text{CHO}}$. The throughput of the CHO-enabled UEs in eq. \eqref{eq:rate_CHO} depends on their own decision, which, in turn, depends on the measured signals and the list of prepared cells it receives from the network. We consider here the standard cell-selection rule, in which the CHO-enabled UE connects to the best-SINR cell among those prepared for it.
\footnote{To avoid often CHOs and ping-pong effects \cite{pp_ho_2023} when finding the highest-SINR prepared cell, it is possible to subtract a cell-specific offset, as happens in the A3 event \cite{3gpp_36_331, 3gpp_38_331, martikainen18}.} 

Finally, we define the load of each cell $j \in \mathcal J$:
\begin{align}
\ell_{j}{(t)}=\sum_{i \in \mathcal I_{\text{THO}}} x_{ij}{(t)}+\sum_{i \in \mathcal I_{\text{CHO}}} y_{ij}{(t)} \color{black}\leq C_j\label{eq:load},
\end{align}
\noindent as the total number of associated users (both THO- and CHO-enabled), \rev{which cannot exceed the upper bound $C_j$, in lieu of the actual spectrum budget capacity.}
Note that for the CHO-enabled users, the cell resources are reserved independently of whether the UEs will actually select the cell, and hence CHOs are likely to induce unnecessary resource waste. This is a key CHO issue that our model tackles.

To streamline the presentation, we define $\bm z_t = (\bm x_t, \bm y_t)$ and the associated set:
\begin{align*}
    \mathcal Z = \Big\{\bm z = (\bm x, \bm y) \ | \ \color{black} \ell_j \leq C_j, \ \forall j \in \mathcal J \color{black}, \ \bm x \in \mathcal X, \ \bm y \in \mathcal Y \Big\}.
\end{align*}

\subsection{Problem Statement}
In our framework, the network controller makes \textit{joint} decisions for the THO- and CHO-enabled users, whose performance \rev{is coupled through the load of each cell}, as can be seen in eq. \eqref{eq:load}.
With this in mind, we introduce the \textit{utility} function the controller needs to maximize:
\begin{align}
\! \!g_t (\bm z_t) \!=\!\! \sum_{j \in \mathcal J} \!\! \Bigg(\! \underbrace{ \sum_{i \in \mathcal I_{\text{THO}}} \!\!\!\! x_{ij}{(t)}\log\frac{c_{ij}{(t)}}{\ell_{j}{(t)}}}_{\substack{\text{explicitly assigned UEs} \\ \text{(can do THO)}}} \!+\!\!\! \underbrace{\sum_{i \in \mathcal I_{\text{CHO}}}\!\!\! y_{ij}{(t)}\log\frac{c_{ij}^\prime(t)}{\ell_{j}{(t)}}}_{\substack{\text{implicitly assigned UEs} \\ \text{(can do CHO)}}} \!\!\Bigg)\!\! \!\label{eq:utility}
\end{align}
\noindent where the two terms define the throughput of the THO-enabled and CHO-enabled UEs. \footnote{\rev{We consider that there is no idle cell, and in case of zero throughput, $c_{ij}(t)+1$ or $c'_{ij}(t)+1$ can be used inside the logarithm.}}
We assume the cell resources are allocated fairly across the users via, e.g., a round robin or a proportional-fair scheduler \cite{tse-scheduler}. The logarithmic transformation balances throughput across users to achieve fairness \cite{andrews-association}; however, other mappings or schedulers can be used as easily.

Using the association decision $x_{ij}{(t)}$ for user $i \in \mathcal I_{\text{THO}}$, the THO is modeled with the change $x_{ij}{(t)} \neq x_{ij}{(t-1)}$. This way, the total number of THOs can be captured with the norm $ \| \bm x_t - \bm x_{t-1}\|$, as in \cite{andrews-association}. Nevertheless, we are interested in the \textit{THO delay (switching) cost} and not merely in the number of THOs. Following our findings in Sec. \ref{sec:data_collection_analysis}, we propose using the weighted norm $\| \bm x_t - \bm x_{t-1}\|_{A_t}$, where $\bm A_t\!=\!\text{diag}( a_{n}{(t)}\!>\!0)$ is a positive definite matrix with its diagonal weights $a_{n}{(t)} \in [0,1]$, $n\!=\color{black} 1, \ldots, \color{black} I_{\text{THO}} \!\cdot J$, \rev{penalizing differently the THOs for each UE-cell pair}; a penalty that may even change over time.\footnote{\rev{E.g., 3G UEs have higher cost for changing cells as they are more prone to HO failures \cite{kalntis_infocom25}. In this way, the model accounts for the UE's HO capabilities.}} 

On the other hand, a CHO is executed from the UE based on the prepared cells by the controller captured through $\bm y_t$: the user's decision in the CHO case is to be assigned to the highest-SINR prepared cell. Clearly, frequent cell preparations and releases may lead to increased signaling \cite{stanczak22, martikainen18}, especially in dense cell deployments (can be up to 161 cells, see Sec. \ref{sec:data_collection_analysis}). Hence, we introduce the \textit{CHO signaling (switching) cost} $\| \bm y_t - \bm y_{t-1}\|_{B_t}$. In this case, $\bm B_t\!=\!\text{diag}( b_{n'}{(t)}\!>\!0, n'\!=\color{black} 1, \ldots, \color{black}I_{\text{CHO}} \!\cdot J)$ is a diagonal positive definite matrix with weights $b_{n'}{(t)}\!\in\! [0,1]$, to penalize differently the preparations and releases of cells for each UE-cell pair\footnote{E.g., the signaling cost might be more detrimental during peak-hours.} in slot $t$. Given that THOs are more prone to failures and delays than CHOs, it holds that $a_{n}{(t)} > b_{n'}{(t)}, \ \forall t \in \mathcal T$.

These THO and CHO switching costs can be combined as: 
\begin{align}
    \| \bm z_t - \bm z_{t-1}\|_{C_t}^2 = & \| \bm x_t - \bm x_{t-1}\|_{A_t}^2 + \| \bm y_t - \bm y_{t-1}\|_{B_t}^2,
\end{align}
\noindent where $\bm C_t$ is the full block diagonal matrix:
\begin{align}
    \bm C_t =
    \begin{bmatrix}
    \bm A_t & 0 \\
    0 & \bm B_t
    \end{bmatrix} \tag{\rev{\theequation}}\addtocounter{equation}{+1},
\end{align}

\noindent and the induced norm and its dual \cite{beck-book}, are:
\begin{align}
    \|\bm{x}_t\|_{A_t}^2\!=\!\!\! \sum_{n=1}^{I_{\text{THO}}\cdot J}\!\! a_{n}{(t)} x_n(t)^2,  \|\bm{x}_t\|_{A_t*}^2\!=\!\!\!\sum_{n=1}^{I_{\text{THO}}\cdot J}\!\!  x_n(t)^2/a_{n}{(t)}, \tag{\rev{\theequation}}\addtocounter{equation}{+1}
\end{align}

\noindent and similarly for $\bm y_t$.

Putting the above together,  the problem that the network controller wishes to address is the following:
\begin{align}
    \mathbb{P}_1: & \max_{\{\bm{z}_t\}_{t}} \sum_{t=1}^T \Big(g_t(\bm z_t) - \| \bm z_t-\bm  z_{t-1}\|_{C_t} \Big) \notag \\ & 
    \textrm{s.t.}
    \quad \bm z_t \in \mathcal Z, \ \forall t \in \mathcal T \notag,
\end{align}
\noindent where $g_t(\bm z_t) -\| \bm z_t-\bm  z_{t-1}\|_{C_t}$ is the \textit{objective} function. 
Next, we explain in detail why $\mathbb{P}_1$ cannot be solved offline.

\subsection{Optimization Challenges}
Solving $\mathbb P_1$ offline is impossible since, at $t=1$, the controller lacks knowledge of the future SINRs of UEs as it has no access to, or influence over, their mobility. Thus, $\mathbb P_1$ must be tackled in an online fashion. In fact, in increasingly-many scenarios, the problem parameters (SINRs, loads, switching costs, etc.) are unknown even at the beginning of each slot; they are only revealed \textit{after} the association and preparation decisions are made. For example, while it may seem the controller could observe the current SINRs before making a decision, in practice, there is a non-negligible delay between the time a UE measures and reports these values and the time the network processes this information. And for fast-moving users or highly-volatile environments, this delay will yield outdated information. This means that we need an online \emph{learning} approach, where both the explicit (for THO) and implicit (for CHO) association decisions are made based on historical UE, network, and environment data, without presuming knowledge of their future values. 

What is more, the signaling and switching overheads depend on successive decision changes: a user that remains in acell or a cell that remains prepared does not incur additional costs. This introduces a memory effect in the optimization, as past decisions influence current decisions. We also stress that the $A_t$-norm and $B_t$-norm imply the costs change over time; however, this cost does not depend on previous decisions. Finally, $\mathbb P_1$ is further compounded by the discreteness of the variables $\bm x_t$ and $\bm y_t$ which prevents the application of off-the-shelf online convex optimization techniques.

Despite these challenges, the goal is to design an online learning algorithm that determines UE-cell associations and preparations and is oblivious to all time-varying and unknown \textit{system conditions}, including future SINR, load, switching costs, and HO delays for each UE-cell pair.

\section{Algorithm Design}\label{sec:algorithm}

To leverage the OCO/SOL toolbox for tackling the problem at hand, the following two main issues need to be addressed: \textit{(i)} the discreteness of decision variables, and \textit{(ii)} the (convex) max operator in the definition of the throughput concerning the CHO-enabled UEs, as can be seen in eqs. \eqref{eq:rate_CHO} and \eqref{eq:utility}. In a nutshell, our solution strategy involves relaxing the discreteness of the decision variables, transforming the resulting continuous problem into a concave one, solving the new concave problem, and then mapping the obtained continuous solution back to the discrete domain via careful rounding.

\subsection{Relaxation \& Transformation}

First, we define the convex hulls $\mathcal X^{\text{c}}\!=\!\text{co}(\mathcal X)$, $\mathcal Y^{\text{c}}\!=\!\text{co}(\mathcal Y)$, $\mathcal Z^{\text{c}}\!=\!\text{co}(\mathcal Z)$ that relax the integrality of the decision variables:
\begin{align*}
	&\mathcal X^{\text{c}}=\Big\{ \bm x \in[0,1]^{I_{\text{THO}}\cdot J}  \  \Big |   \sum_{j\in \mathcal J}{x}_{ij} = 1,  i\in \mathcal I_{\text{THO}} \Big\}, \\
    & \mathcal Y^{\text{c}}=\Big\{ \bm y \in [0,1]^{I_{\text{CHO}}\cdot J}  \,\,\,  \Big |   1 \leq \sum_{j\in \mathcal J} {y}_{ij} \leq b_{i},  i\in \mathcal I_{\text{CHO}} \Big\}, \\
    & \mathcal Z^{\text{c}}=\Big\{ \bm z = (\bm x, \bm y) \ | \ \color{black} \ell_j \leq C_j, \ \forall j \in \mathcal J, \ \color{black} \bm x \in \mathcal X^{\text{c}}, \ \bm y \in \mathcal Y^{\text{c}} \Big\}.
\end{align*}

\noindent Then, using the properties of logarithm, eq. \eqref{eq:utility} becomes:
    \begin{align}
    g_t(\bm z_t) &= \underbrace{\sum_{j \in \mathcal J} \sum_{i \in \mathcal I_{\text{THO}}} \!\!\!\! x_{ij}(t)\log{c_{ij}(t)}}_{g^{\text{THO}}_t(\bm x_t)} \!+\!
    \underbrace{\sum_{j \in \mathcal J}\sum_{i \in \mathcal I_{\text{CHO}}} \!\!\!\! y_{ij}(t)\log{c_{ij}^{\prime}(t)}}_{g^{\text{CHO}}_t(\bm y_t)} \notag
    \\ & - 
    \underbrace{\sum_{j \in \mathcal J} \ell_{j}(t) \log\ell_{j}(t)}_{g^{\text{load}}_t(\bm x_t, \ \bm y_t)} \tag{\rev{\theequation}}\addtocounter{equation}{+1}\label{eq:utility_open},
    \end{align}
    \noindent where the throughput for the THO-enabled UEs, $g^{\text{THO}}_t(.)$, is linear in $\bm x_t \!\in\! \mathcal X^\text{c}$, and the load penalty, $-g^{\text{load}}_t(.)$, describes the entropy; thus, both are \textit{concave}. On the other hand, the throughput for the CHO-enabled UEs, $g^{\text{CHO}}_t(.)$, is not concave even when $\bm y_t \!\in\! \mathcal Y^\text{c}$, due to the influence of preparation decisions $\bm y_t$ on $c_{ij}^{\prime}(t)$ \rev{through the max operator}, see eq. \eqref{eq:rate_CHO}. To illustrate this more clearly, we can equivalently write this part as follows:
    \begin{align}
        g^{\text{CHO}}_t(\bm y_t) \triangleq \sum_{i\in\mathcal I_{\mathrm{CHO}}} \color{black}\max_{j:\,y_{ij}(t)=1}\,\log c_{ij}(t),\label{eq:max_function}
    \end{align}
\noindent where this transformation achieves the same/desired behavior: the CHO-enabled UE obtains throughput from the cell with the best SINR \textit{among the prepared ones}.
Given that the max operator of eq. \eqref{eq:max_function} is nonsmooth and convex, \rev{we introduce a linear masking mechanism that leads to a \textit{concave} surrogate, inspired by the standard LogSumExp function \cite{boyd} used in deep neural networks (DNNs):}

\begin{align}
    \tilde{g}^{\text{CHO}}_t(\bm y_t) \triangleq \sum_{i\in\mathcal I_{\mathrm{CHO}}}\frac{1}{\alpha}\log\!\color{black}\Bigg(\sum_{j \in \mathcal J} y_{ij}(t)\,c_{ij}(t)^{\alpha}\Bigg),\label{eq:max_function_approx}
\end{align}

\noindent where $\alpha > 0$ controls the tightness of the approximation.

\begin{lemma}[\color{black}{Approximation of CHO Throughput}]\label{lemma_relaxation}

    For $\bm y_t \in \mathcal Y^{\text{c}}$ and $t \in \mathcal T$, 
    \[
    \lim_{\alpha \to \infty} \tilde{g}^{\text{CHO}}_t(\bm y_t) = g^{\text{CHO}}_t(\bm y_t).
    \]
    
\end{lemma}
\begin{proof}

\noindent \rev{Let $d^*_{i}(t) = \max_{j:\,y_{ij}(t)=1}\,\log c_{ij}(t)$ and $d_{ij}(t) = \log c_{ij}(t)$. Then, eq. \eqref{eq:max_function_approx} becomes:
\begin{align}
     &\tilde{g}^{\text{CHO}}_t(\bm y_t) =
     \sum_{i\in\mathcal I_{\mathrm{CHO}}}\!\!\frac{1}{\alpha}\log\!\Bigg(\sum_{j \in \mathcal J} y_{ij}(t) e^{\alpha d_{ij}(t)} \Bigg) = \notag \\
     & \sum_{i\in\mathcal I_{\mathrm{CHO}}}\!\!\frac{1}{\alpha} \log\!\, \!\Bigg(e^{\alpha d^*_{i}(t)} \bigg( 1 + \!\!\!\!\!\!\!\sum_{j:\, d_{ij}(t) \neq d^*_{i}(t)} \!\!\!\!\!\!\!\!y_{ij}(t) e^{\alpha\big(d_{ij}(t) - d^*_{i}(t)\big)} \bigg) \Bigg) = \notag \\ 
     & \sum_{i\in\mathcal I_{\mathrm{CHO}}} \!\!\! \Bigg( d^*_{i}(t) + \frac{1}{\alpha} \log\!\, \!\bigg( 1 + \!\!\!\!\!\!\!\sum_{j:\, d_{ij}(t) \neq d^*_{i}(t)} \!\!\!\!\!\!\!\!y_{ij}(t) e^{\alpha\big(d_{ij}(t) - d^*_{i}(t)\big)} \bigg) \Bigg), \notag
\end{align}}

\rev{\noindent which converges to the max, namely, $d^*_{i}(t)$, since $d_{ij}(t)<d^*_{i}(t)$ by definition; and then $e^{\alpha(d_{ij}(t) - d^*_{i}(t))}\to 0$ as $\alpha\to\infty$.}

\end{proof}

\rev{Given that from Lemma \ref{lemma_relaxation}, eq.  \eqref{eq:max_function_approx} approximates eq. \eqref{eq:max_function}, the transformed utility:
\begin{align}
        \tilde{g}_t(\bm z_t) = g^{\text{THO}}_t(\bm x_t) 
        + \tilde{g}^{\text{CHO}}_t(\bm y_t) 
        - g^{\text{load}}_t(\bm x_t, \bm y_t) \label{eq:approx_utility}
    \end{align}
is concave, as desired. 
}

\subsection{The Meta-Learning Approach}

We approach
$\mathbb{P}_1$ as a \emph{smoothed online learning} problem and solve it using a \emph{meta-learning} approach based on the \emph{experts} framework \cite{warmuth-experts, hazan-meta}, \rev{as can be seen Figure \ref{fig:algo_mechanism}}. This is particularly well-suited to our setting because switching costs are incorporated into the formulation, and system conditions can vary significantly across time slots due to, e.g., unknown user mobility. As a result, deploying a single learner with a fixed learning rate may perform well in some regimes but poorly in others.

We deploy a set $\mathcal K$ of $K$ learning agents, called \emph{experts}\rev{, or simply \emph{learners}}, each with a different learning step/rate $\theta=(\theta_k, k\in\mathcal K)$ that is applied to online gradient ascent (OGA) \cite{zinkevich}. \rev{An expert with a larger learning rate puts more emphasis on the latest gradient and hence adapts more quickly to rapid changes in SINRs and signaling costs}, making it more suitable for highly volatile scenarios. \rev{In contrast, experts with a smaller learning rate update their decisions more carefully (do not move fast from previous decisions) and hence are less affected by parameter fluctuations. This makes them better suited for relatively stationary scenarios as they induce less handovers (decision changes)}. 

\rev{Given that the volatility of the environment and UEs' mobility patterns are unknown in advance, we use a large enough set $\mathcal K$ of experts, each tuned with a different learning rate. The rationale is that, with careful selection of the rates (number of experts and steps), at least one expert is guaranteed to perform well for the encountered scenario. We identify the best expert for each scenario at runtime by using a meta-learning algorithm, which implements Hedge \cite{freund_hedge} (and not OGA as the experts) and learns how much weight to put on each expert's proposal. These weights are dynamically updated based on observed expert performance.} 

\begin{figure}[!t]
    \centering
    \includegraphics[width=0.95\columnwidth]{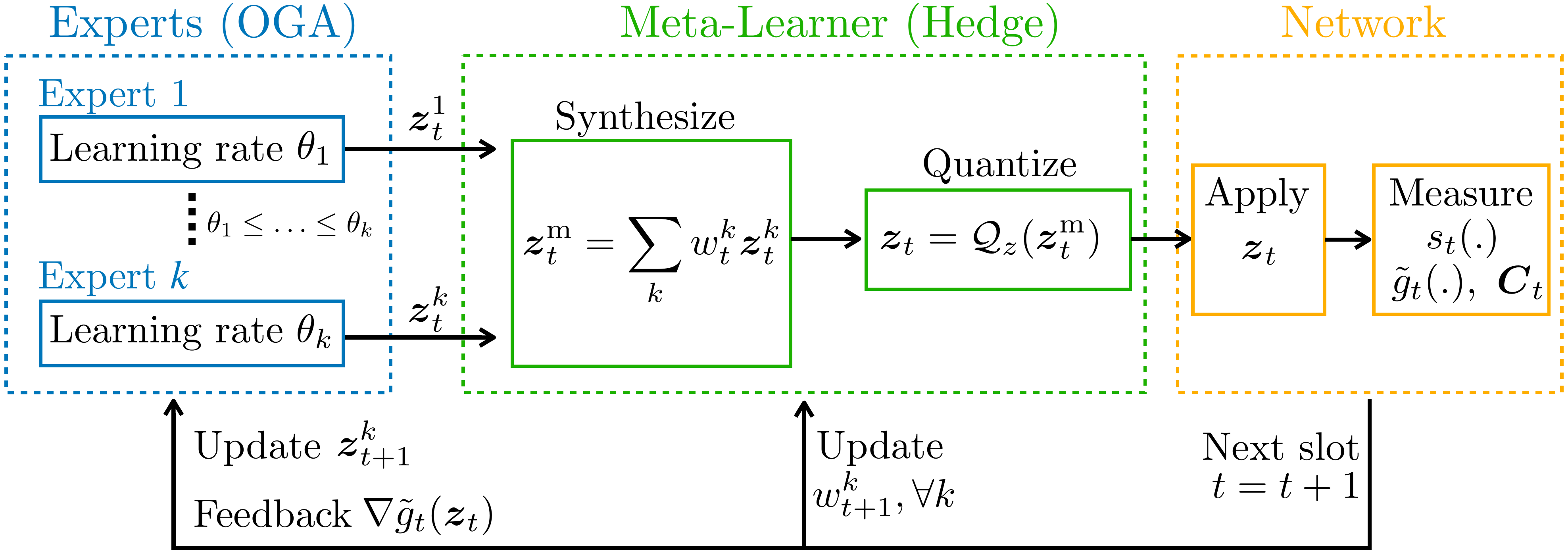}
    \caption{\rev{Learning mechanism. \textit{(i)} Each expert uses a different learning rate and proposes a HO policy to the meta-learner (blue box). \textit{(ii)} The meta-learner combines the proposals, using the experts' weights, and implements the result after discretizing it (green box). \textit{(iii)} The mechanism observes the performance and assesses how good each expert's proposal was (orange box). \textit{(iv)} The experts adapt their decisions and the meta-learner adapts their weights.}}
    \label{fig:algo_mechanism}
\end{figure}

Formally, at the beginning of each slot $t=1, ..., T$, each expert $k \in \mathcal K$ shares its suggestion $\bm{z}_t^k = (\bm x_t^k, \bm y_t^k)$, where we initialize $z_0^k =0$ for all $k$. The meta-learner combines these suggestions into a single decision $\bm z_t^\text{m} = (\bm x_t^\text{m}, \bm y_t^\text{m})$ as follows:
\begin{align}
    \bm z_t^\text{m}=\sum_{k\in\mathcal K}w_t^k\bm{z}_t^k,\label{eq:meta-mix}
\end{align}
\noindent where $\bm{w}_t\!=\!(w_t^k, k\in \mathcal K)$, with $\bm w_t^\top \bm{1}_K\!=\!1$, are the meta-learner's adaptive weights for each expert. The goal is to assign higher weights to experts that perform better.

Since the decision $\bm z_t^\text{m}$ of the meta-learner takes continuous values due to $\bm{w}_t$, we apply a quantization function $Q_{\mathcal Z}$ to project it back to a valid (i.e., implementable) discrete decision $\bm z_t$. We require that the quantization is unbiased, i.e., $\mathbb E[\bm z_t]\!=\!\bm z_t^\text{m}$, which holds by: $(i)$ picking a cell $j$ to assign a THO-enabled UE $i \in \mathcal I_{\text{THO}}$ with probabilities $\bm{x}_i^\text{m}(t)$, creating $\bm x_t \in \mathcal X$, and $(ii)$ deciding whether to prepare a cell $j$ for a CHO-enabled UE $i \in \mathcal I_{\text{CHO}}$ through sampling from a Bernoulli distribution with probabilities $\bm{y}_{ij}^\text{m}(t)$, creating $\bm y_t \in \mathcal Y$.

Once the decision $\bm z_t$ is implemented, the controller observes the SINR values $s_{ij}(t)$ for all UEs and cells, and calculates $\tilde{g}_t(.)$ through eq. \eqref{eq:approx_utility}, as well as the switching costs (by observing $C_t$). We pad with zeros the entries $s_{ij}(t)$ for which no SINR is received, because these cells are unreachable. As a next step, the meta-learner evaluates the decision $\bm z_t^k$ of each expert using the surrogate (i.e., partially linearized) loss:
\begin{align}\label{eq:loss-function}
l_t(\bm z_t^k)=-\langle{\nabla \tilde{g}_t(\bm z_t) },{\bm z_t^k - \bm z_t}\rangle +  \| \bm z_t-\bm  z_{t-1}\|_{C_t}
\end{align}
\noindent and updates its weights using the step $\eta$ (defined later) as:
\begin{align}\label{experts-weights-update}
    w_{t+1}^k=\frac{w_{t}^ke^{-\eta l_t(\bm z_t^k)}}{\sum_{k\in\mathcal K} w_{t}^ke^{-\eta l_t(\bm z_t^k)} }.
\end{align}
\noindent Lastly, it sends the gradient  $\nabla g_t(\bm z_t)$ in the implementable decision $\bm z_t$ to all experts, and they update their choices via OGA:
\begin{align}\label{expert-update}
	\bm{z}_{t+1}^k=\Pi_{\mathcal Z}\Big( \bm{z}_t^k + \theta_k\nabla \tilde{g}_t(\bm z_t)\Big).
\end{align}

\noindent The projection $\Pi_{\mathcal Z}(.)$ ensures that the proposed decisions $\bm{z}_{t+1}^k$ lie in the feasible space,
namely, $\bm z_{t+1}^k \in \mathcal Z$. A summary of the steps can be seen in Algorithm \ref{alg:contra} (\texttt{CONTRA}).

\begin{algorithm}[t] 
	{\small{	
			\nl \textbf{Required}: Step $\eta$ for meta-learner and $\{\theta_k\}_{k\in \mathcal K}$ for experts \\ 
			\nl \textbf{Initialize}: Sort $\theta_1\leq \theta_2\leq \ldots \leq \theta_K$ and set ${w}_1^{k}=\frac{1+1/K}{k(k+1)}$\\
			\nl \For{ $t=1,2,\ldots, T$  }{
				\nl Each expert $k\in \mathcal K$ shares its decision $\bm z_t^k$\\
				\nl  Combine all decisions $\bm z_t^k$ into $\bm z_t^\text{m}$ using \eqref{eq:meta-mix}\\
				\nl Create an implementable discrete decision $\bm{z}_t=Q_{\mathcal Z}(\bm z_t^\text{m})$\\	
                    \nl Observe SINRs, $\tilde{g}_t(.)$ and switching costs $C_t$\\
				\nl Update the weights of experts using \eqref{eq:loss-function} and \eqref{experts-weights-update} \\							
				\nl Send $\nabla g_t(\bm z_t)$ to each expert \\
                    \nl Each expert updates its decision using \eqref{expert-update} \\							
			}
			\caption{{\small{\underline{\smash{CON}}ditional-\underline{\smash{TRA}}ditional HOs (\texttt{CONTRA})}}}\label{alg:contra}
	}}
	\setlength{\intextsep}{0pt} 
\end{algorithm}

\subsection{Performance Guarantees \& Complexity}
To assess the performance of our algorithm, we compare it against a powerful benchmark (i.e., \textit{oracle}) that has full a priori knowledge of the system conditions (i.e., SINR, load, costs etc) and can choose the best possible sequence of decisions over the entire horizon to maximize the cumulative objective of $\mathbb P_1$. This serves as a highly competitive benchmark, exceeding both static (single best solution for all time slots) and dynamic ones (best solution for each time slot independently) \cite{lattimore, hazan-book}. 

For that, we leverage the \emph{Expected Dynamic Regret} \cite{zhang-smoothed-ol}:
\begin{align}
\mathbb{E}\left[{\mathcal R}_T\right] & \triangleq \sum_{t=1}^T\Big( \tilde g_t(\bm{z}_t^\star) - \| \bm z_t^\star-\bm  z_{t-1}^\star\|_{C_t}\Big) \notag \\ &
- \sum_{t=1}^T \mathbb{E}\Big[ \tilde g_t(\bm{z}_t) - \| \bm z_t-\bm  z_{t-1}\|_{C_t}\Big], \label{regret-metric}
\end{align}
where $\{\bm{z}_t\}_t$ and $\{\bm z_t^{\star}\}_t$ are the algorithm's and powerful-oracle's decisions, respectively. The latter is the solution of $\mathbb{P}_1$, and the expectation captures the randomization in \texttt{CONTRA}'s decisions due to quantization $Q_{\mathcal Z}$. Our goal is to design an algorithm that ensures this gap diminishes with time, $\lim_{T\rightarrow \infty} \mathbb{E}[{\mathcal R}_T]/T\!=\!0$ for any possible benchmark sequence $\{\bm z_t^{\star}\}_t$. It is important to highlight that achieving close-to-zero expected dynamic regret is particularly challenging. 

Ideally, designing an effective algorithm in such settings requires prior knowledge of how much optimal decisions can vary over time, commonly characterized by the \textit{path length}:
\begin{align}
    P_T = \sum_{t=1}^T \|\bm{z}_t^{\star}-\bm{z}_{t-1}^{\star}\|_{C_t}.
\end{align}

\noindent However, this information is not available in practice. To overcome this limitation, we have relied on a meta-learning algorithm that adaptively learns from a pool of experts with different learning rates, each tailored to perform well under different system conditions. To bound the expected dynamic regret, we first prove the following lemma.
\begin{lemma}[Bound of Domain and Gradients]\label{lemma_dom_grads}
By considering: $a_{n}{(t)} \leq a_{\max}, \ b_{n'}{(t)} \leq b_{\max}, \ c_n'(t) \leq c_n(t) \leq c_{\max} $ and $M \!=\! \max\left\{ \log c_{\max} \!-\! 1, \ \log I + 1 \right\}$, with $t \!\in\! \mathcal T$, it holds that:

\noindent $\bullet$ $\| \bm z - \bm z'\|_2 \leq \sqrt{2 \, I_{\text{THO}}} + \sqrt{I_{\text{CHO}} \, (J-1)} \triangleq D$,

\noindent $\bullet$ $\| \bm z - \bm z'\|_{C_t} \leq  \sqrt{2 \, I_{\text{THO}} \ a_{\max}} + \sqrt{I_{\text{CHO}} \, (J-1) \ b_{\max}}  \triangleq D_\text{C}$,

\noindent $\bullet$ $\| \bm z \!-\! \bm z'\|_{C_t*} \! \leq \!  \sqrt{2\,I_{\text{THO}}/a_{\max}} +\! \sqrt{I_{\text{CHO}} \, (J\!-\!1)/b_{\max}}  \triangleq D_\text{C*}$,

\noindent $\bullet$ $\|\nabla \tilde g_t(\bm z)\|_2 \leq M\sqrt{I \,J} \triangleq G$,

\noindent $\bullet$ $\|\nabla \tilde g_t(\bm z)\|_{C_t} \leq M\sqrt{(I_{\text{THO}}\ a_{\max}+I_{\text{CHO}}\ b_{\max}) \, J} \triangleq G_{\text{C}}$.
\end{lemma}

\begin{proof}
We provide the full proof in the Appendix.
\end{proof}

The main result for the performance guarantee of \texttt{CONTRA} is presented in the sequel.

\begin{lemma}[Performance Analysis / Optimality Guarantee]\label{lemma_main} Using the parameters:
\begin{align}
    & \bullet \ \  K=\left\lceil\log_2\sqrt{1+\!2T} \right \rceil+1, \label{eq:num_experts}\tag{\rev{\theequation}}\addtocounter{equation}{+1} \\
    &\bullet \ \ \theta_k=2^{k-1}\sqrt{\frac{D_\text{C}^2}{T(G^2+2G_\text{C})}},  \  k=1,\ldots,K, \label{eq:lr_experts} \tag{\rev{\theequation}}\addtocounter{equation}{+1} \\
    & \bullet \ \ \eta=1/\sqrt{T\nu} \ \text{ with } \tag{\rev{\theequation}}\addtocounter{equation}{+1}
    \\ & \quad \,\,\,\,  \nu\triangleq (D_\text{C}\!+\!1/8)(GD\!+\!2D_\text{C})^2, \tag{\rev{\theequation}}\addtocounter{equation}{+1}
\end{align}

\noindent the discrete decisions $\{\bm z_t\}_t$ of our algorithm \texttt{CONTRA} ensure: 
\begin{align}
\mathbb{E}\left[{\mathcal R}_T \right]\! \leq \! \notag \sqrt{T}\Big( & \sqrt{\nu}\left(1+\ln(1/w_1^k)\right) + \\ & \sqrt{(G^2+2G_\text{C})(D_\text{C}^2+2D_{\text{C}*}P_T)} \ \Big) \notag + \\ & \hspace{-1.9cm} T\Big(G\!+\!\sqrt{a_{\max}\!+\!b_{\max}}\Big)\sqrt{I_{\text{CHO}}J/4 + I_{\text{THO}}(1\!-\!1/J)} \label{eq:lemma_res},
\end{align}
\end{lemma}

\begin{proof}[Proof]
We provide the full proof in the Appendix.

\end{proof}

With Lemma \ref{lemma_main}, we bound the expected regret of the implementable decisions, in contrast to previous works \cite{andrews-association, kelleler-jsac23}. Even though the continuous preparation variables achieve sublinear dynamic regret, the discretization introduces an unavoidable non-diminishing error.
As shown in Sec. \ref{sec:evaluation}, this error is small in practical scenarios, and the algorithm converges towards the optimal solutions.

\rev{\noindent\textbf{Computational Complexity.} The primary source of complexity in our algorithm arises from the projection step performed by each learner $k \in \mathcal K$ during OGA, as shown in eq. \eqref{expert-update}; all other computations are executed in constant time $\mathcal{O}(1)$. The projection difficulty depends on \textit{(i)} the structure of the feasible set, which essentially depends on the constraints of our problem, and \textit{(ii)} the size of the problem, which in practice corresponds to the area monitored by each controller (i.e., the number of UEs and cells). 
We note that for convex sets in general, there might be no closed-form solution, with the projection requiring $\mathcal O(d^3)$, where $d$ is the set's dimension.}

\rev{Our current implementation indeed computes the projections directly by solving the concave subproblems, which have been found sufficient for the studied problems. Specifically, and despite relying on a general-purpose solver rather than specialized projection routines and a modest hardware (default) setup of a MacBook Pro equipped with an Apple M1 chip (8-core CPU), the per-slot inference (execution) time remains below the near-RT threshold required in O-RAN operation when up to 150 UEs and 25 cells are considered. 

More precisely, the inference time increases from 18 ms ($I=1$ user) to 62 ms ($I=150$ users) when $J=5$ cells. For $J=15$ ($J=25$), it begins near 18 ms (18 ms) for $I=1$ user and reaches up to 384 ms (1 s) for $I=150$ users. It is important to note that the exact runtime may vary depending on the specifications of different machines. For larger search spaces, one can employ tailored projection algorithms, such as \cite{shalev_proj, paschos_oco_caching}, and/or execute the algorithm on GPUs or high-performance computing servers.}

\section{Dynamic Handover Type Selection }\label{sec:model_extension}

\noindent\rev{\textbf{Problem Formulation.}} 
In contrast to the framework of Sec. \ref{sec:system_model}, we now consider a more flexible and dynamic formulation in which users are not a priori assigned to a specific HO type. Instead, the network controller can decide on-the-fly (i.e., at each time slot $t$), whether a UE should be THO- or CHO-enabled, based on current system conditions and performance trade-offs, enabling a more responsive and efficient HO strategy, as envisioned by 6G and NextG \cite{chen_5G6G}. Specifically, we allow each UE to be in either HO mode and decide this dynamically: when a UE moves fast, it may benefit from CHO due to the rapid signal fluctuations that lead to higher THO failures and delays, while it might find THO suitable when it moves slower.

In this case, the controller problem becomes: 
\begin{align}
\mathbb{P}_2: & \max_{\{\bm{z}_t\}_{t}} \sum_{t=1}^T \Big(g_t(\bm z_t) - \ \| \bm z_t-\bm  z_{t-1}\|_{C_t} \Big) \notag \\ & 
\textrm{s.t.} \quad
\sum_{j \in \mathcal J} x_{ij}(t) \leq 1, \label{eq:THO_association_real}  \quad \forall i \in \mathcal I, t\in \mathcal T
\\ 
& \quad\quad y_{ik}(t) \leq 1 \!-\! \sum_{j \in \mathcal J} x_{ij}(t), \label{eq:CHO_association_real}  \ \forall i \in \mathcal I, k \in \mathcal J, t\in\mathcal T
\\ 
& \quad\quad\sum_{j \in \mathcal J} y_{ij}(t)\leq b_i, \label{eq:max_prepared_cells_real} \quad \forall i \in \mathcal I, t \in\mathcal T
\\ 
& \quad\quad \sum_{j \in \mathcal J} x_{ij}(t) + \sum_{j \in \mathcal J} y_{ij}(t) \geq 1, \label{eq:either_CHO_or_THO} \quad \forall i \in \mathcal I, t \in \mathcal T
\\ 
& \color{black} \quad\quad \sum_{i \in \mathcal I} x_{ij}{(t)}\!+\!\sum_{i \in \mathcal I} y_{ij}{(t)} \color{black}\leq C_j, \label{eq:load_capac_constr} \ \forall j \in \mathcal J, \forall t \in \mathcal T 
\\ 
& \quad\quad \color{black} \bm z_t = (\bm x_t, \bm y_t) \in \{0,1\}^{I\cdot J}\!\times\!\{0,1\}^{I\cdot J}, \label{eq:binary_decisions_real} \ \forall t \in \mathcal T.
\end{align}

\noindent Eq. \eqref{eq:THO_association_real} allows a UE $i \in \mathcal I$ to be THO-enabled (i.e., explicitly assigned to a cell if $x_{ij}(t) \!=\! 1$), or CHO-enabled (i.e., implicitly assigned through CHOs to a cell, otherwise). Eq. \eqref{eq:CHO_association_real} complements the previous one, ensuring that a UE can be explicitly or implicitly assigned to a cell, but not both; eq. \eqref{eq:max_prepared_cells_real} prevents each CHO-enabled UE from having more than $b_i$ prepared cells, similarly to Sec. \ref{sec:system_model}. Eq. \eqref{eq:either_CHO_or_THO} ensures that a UE will be explicitly or implicitly assigned, preventing blocking completely some UEs\rev{, and eq. \eqref{eq:load_capac_constr} enforces each cell's capacity constraints using the total number of users.} Lastly, eq. \eqref{eq:binary_decisions_real} defines the binary decisions.

\noindent\rev{\textbf{Performance Analysis.}} 
The analysis and proofs follow nearly verbatim those of Sec. \ref{sec:system_model}, with the modification of the sampling technique $Q_{\mathcal Z}$. \rev{We highlight that the analysis of Sec. \ref{sec:system_model} relies on fixed sets $\mathcal I_{\text{THO}}$ and $\mathcal I_{\text{CHO}}$, and therefore,} the unbiased estimator defined earlier cannot be used in the case of dynamic HO type selection, where $\bm x_t$ and $\bm y_t$ are coupled \rev{with the controller deciding the HO type of each UE in every slot.} For that reason, we create a new unbiased estimator. 

\rev{Let $\bm z_t^{\text{m}}=\big(\bm x_t^{\text{m}},\bm y_t^{\text{m}}\big)$ denote the \emph{continuous} meta-learner output at slot $t$, where $\bm x_t^\text{m} \in [0,1]^{I\cdot J}$ and $\bm y_t^\text{m} \in [0,1]^{I\cdot J}$. The definition of $\bm z_t^{\text{m}}$ resembles eq. \eqref{eq:meta-mix}, with the distinction that in that earlier formulation, the users were partitioned into two disjoint sets $\mathcal I_{\text{THO}}$ and $\mathcal I_{\text{CHO}}$.}

\rev{To construct the implementable binary decision $\bm z_t = (\bm x_t, \bm y_t)$, we start by} defining for each $i \in \mathcal I$ the probability to be explicitly assigned (THO-enabled) \rev{at slot $t$}:
\begin{align}
    \color{black} \pi_i(t) \;\triangleq\; \color{black} \sum_{j\in\mathcal J} x^{\text{\textcolor{black}{m}}}_{ij}(t),
\end{align}

\noindent where eq. \eqref{eq:THO_association_real} ensures that $\pi_i(t) \in [0,1]$. To obtain an unbiased estimator, i.e., $\mathbb E[\bm z_t]\!=\!\bm z_t^{\text{\textcolor{black}{m}}}$, we sample:
\begin{align}
    {z}_{i}(t) \sim \text{Bernoulli}\big(\pi_i(t)\big) \tag{\rev{\theequation}}\addtocounter{equation}{+1},
\end{align}
\noindent and if ${z}_{i}(t)=1$, the UE $i \in \mathcal I$ is THO-enabled at slot $t$, else CHO-enabled. \rev{Lastly, the associations/preparations are decided with the normalized quantities: ${x}_{ij}(t) = {x_{ij}^{\text{\textcolor{black}{m}}}(t)}/{\pi_i(t)}$ and ${y}_{ij}(t) = {y_{ij}^{\text{\textcolor{black}{m}}}(t)}/{(1-\pi_i(t))}.$}

From the law of total expectation, and as $ x_{ij}(t)\!=\!0$ when $ z_i(t)\!=\!0$, we get $\mathbb E[x_{ij}(t)]\!=\!x_{ij}^{\text{\textcolor{black}{m}}}(t)$ and similarly $\mathbb E[y_{ij}(t)]\!=\! y_{ij}^{\text{\textcolor{black}{m}}}(t)$. This concludes that indeed, $\mathbb E[\bm z_t]\!=\!\bm z_t^{\text{\textcolor{black}{m}}}$.

\section{Performance Evaluation}\label{sec:evaluation}

%
%
%
%
%
%

\color{black}
We compare \texttt{CONTRA} with various \textit{(i)} 3GPP-compliant, threshold-based THO and CHO algorithms, as well as \textit{(ii)} RL benchmarks. The former are the algorithms currently used by MNOs and antenna vendors \cite{3gpp_38_300, 3gpp_38_331}, solving either solely the THO or CHO problem (i.e., not jointly). Given that these benchmarks treat THO and CHO as independent procedures, we adapt and extend \textit{three} existing foundational \textit{RL} algorithms from the literature that were originally designed for other tasks \cite{reinforce_williams, kelleler-jsac23, ppo_scc}, such as only THO optimization. In this way, it becomes possible to compare the proposed algorithm with more advanced (i.e., not relying on pre-defined thresholds) benchmarks for precisely the problem at hand, namely, the joint optimization of traditional and conditional handovers.

\begin{figure}[t]
    \centering
    \includegraphics[width=0.6\columnwidth]{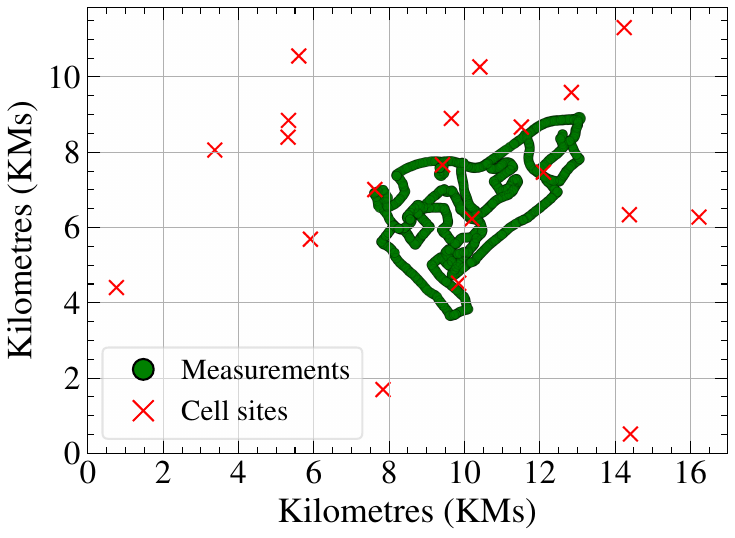}
    \caption{Cell sites (red) and SINR measurement locations (green) used for evaluation.}
    \label{fig:real_sinr_map_anon}
\end{figure}

The algorithms are evaluated in two main scenarios: \textit{(i)} \textit{volatile} SINR scenario, in which $s_{ij}(t)$ fluctuates every 10 slots within the range of [0, 30]dB \cite{3gpp_36_133}, reflecting a wide spectrum of signal conditions (i.e., from poor to excellent values), and \textit{(ii)} real SINR scenario from crowdsourced countrywide measurements. For the latter, we focus on an area of approximately 160 sq. kms, where 73,303 signal measurements in a second granularity are taken for 100 cells (20 cell sites). These cell sites are shown in Figure \ref{fig:real_sinr_map_anon}, marked in red, and the measurement areas are indicated in green. This dataset presents a highly competitive scenario as SINRs can vary arbitrarily and occasionally exhibit adversarial behavior (see Sec. \ref{sec:data_collection_analysis}). Given that these measurements are anonymized, we consider \rev{100} fast-moving UEs, with random velocities ranging from 20 to 28 m/s and variances in [0, 5], placed at $t=1$ in a random location with measurement (green). We adopt the Gauss-Markov mobility model \cite{gauss_markov}, in line with prior works \cite{andrews-association}, with 0.9 randomness parameter to find the location of each UE in each subsequent slot $t=2, ..., T$, for $T=1$k. To incorporate the signal measurements, we map each new location to the nearest location with available measurements.

%
%
%
%
%
%

\rev{\subsection{Understanding the Proposed Algorithm}}

First and foremost, it is imperative to showcase the differences between the static (Sec. \ref{sec:system_model}) and dynamic (Sec. \ref{sec:model_extension}) HO type selection. For that reason, we compare four variants of \texttt{CONTRA}: \textit{(i)} dynamic HO type selection, where the controller decides if and when each user should operate in THO or CHO mode according to our framework, and three variants of static HO type selection, namely, \textit{(i)} a random split with half of the users predefined as THO-enabled and the other half as CHO-enabled, as well as a configuration in which all users are \textit{(ii)} THO-enabled, and \textit{(iii)} CHO-enabled.

\begin{figure}[t]
    \centering
    \subfigure[]{\includegraphics[width=0.48\linewidth]{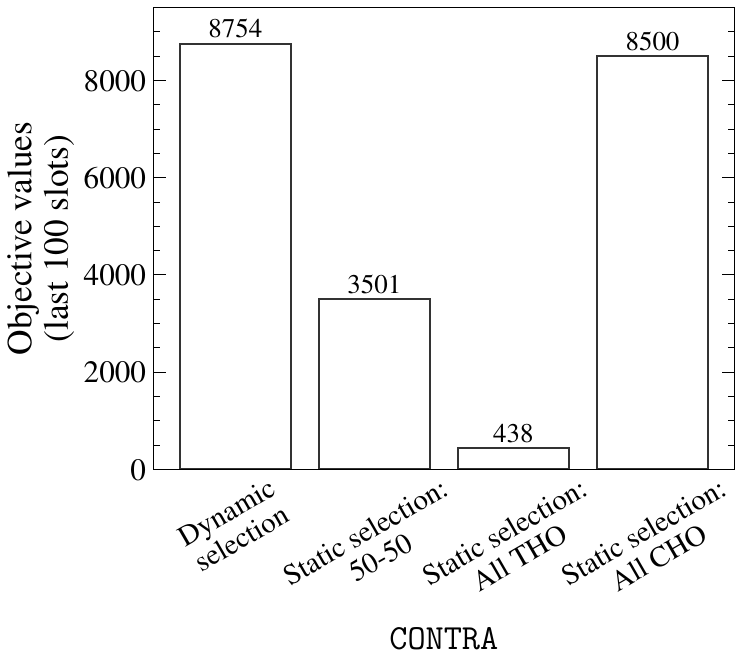}\label{fig:static_vs_dynamic_model_objective}}
    \subfigure[]    {\includegraphics[width=0.47\linewidth]{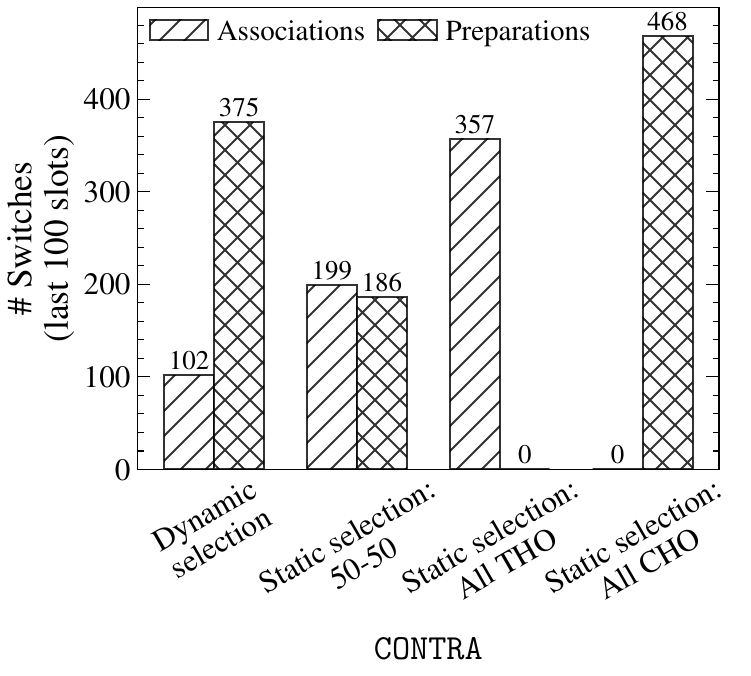}\label{fig:static_vs_dynamic_model_switches}}
    \caption{\rev{Real SINR scenario with $I=100$ UEs. Over the last 100 slots of the $T=1$k simulation, (a) shows the objective function values (i.e., throughput and switching costs), and (b) the number of switches in associations (i.e., $\bm x_t$) and preparations (i.e., $\bm y_t$), for the dynamic and three variants of static HO type selection of \texttt{CONTRA}.
    }} 
    \vspace{-5mm}
    \label{fig:static_vs_dynamic_HOpolicy}
\end{figure}

In detail, Figure \ref{fig:static_vs_dynamic_HOpolicy} compares the performance of these four \texttt{CONTRA} variants over the last 100 slots of the real-SINR scenario. When all users are in THO mode, the objective function values are the lowest (i.e., 438), primarily due to the large number (357) and high cost of THOs. The 50–50 split between THO- and CHO-enabled users improves performance by nearly $\times$8 compared to the all-THO configuration; yet it remains roughly $\times$2.5 lower than both the automatic mode-selection and all-CHO cases. The automatic mode achieves higher performance ($\approx3\%$) compared to the all-CHO configuration, as it executes fewer CHO preparations, even though it requires 102 THOs. From here onward, \texttt{CONTRA} is assessed in the dynamic HO formulation.

Secondly, we focus on the trade-off between throughput and the two switching cost terms to understand the importance of the latter for the studied problem. Figure \ref{fig:sw_costs_g} shows the utility function and the number of association and preparation changes when switching costs are \textit{(i)} excluded, and \textit{(ii)} taken into account, for the real-SINR scenario. Focusing on the last 250 slots of $T=1$k, we note that taking into account the switching cost reduces the number of associations (and thus, THOs) and preparation changes by approximately 77\%, while achieving a 2.7\% higher throughput. Hence, it becomes evident that solely maximizing the utility/throughput (i.e., without considering the switches) may not be the optimal strategy, given that comparable, or even better, performance can be achieved with fewer switches, thereby reducing signaling overhead, resource utilization, and HO delays.

\begin{figure}[t]
    \centering
    \subfigure[]{\includegraphics[width=0.38\linewidth]{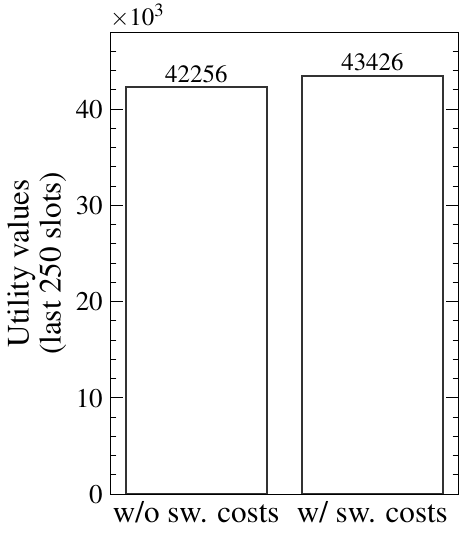}\label{fig:thr_sw_costs1}}
    \subfigure[]{\includegraphics[width=0.58\linewidth]{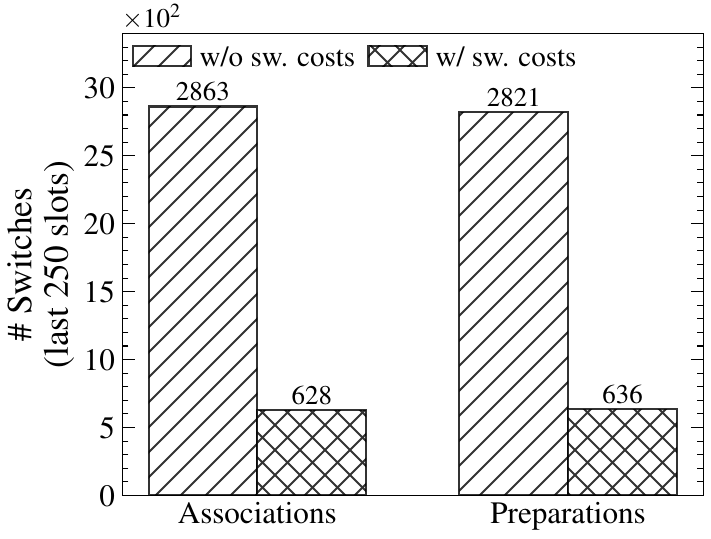}\label{fig:thr_sw_costs2}}
    \caption{\rev{Real SINR scenario with $I=100$ UEs. Over the last 250 slots of the $T=1$k simulation, (a) shows the utility function values (i.e., $\tilde{g}_t$), and (b) the number of switches in associations (i.e., $\bm x_t$) and preparations (i.e., $\bm y_t$), for ignoring and taking into account the switching costs.}} 
    \vspace{-5mm}
    \label{fig:sw_costs_g}
\end{figure}

Thirdly, to show the importance of combining learners with different rates, \rev{Figure \ref{fig:tilde_g_real_stationary}} shows the attained throughput $\tilde g(.)$ in a case where the SINRs are \textit{stationary} (changing only every 200 slots) and real (changing significantly between slots). Given that $T=1$k, eq. \eqref{eq:num_experts} determines that 7 experts will be used, with learning rates 0.0135, 0.0271, 0.0542, 0.1083, 0.2167, 0.4334, and 0.8668 (combining eq. \eqref{eq:lr_experts} and Lemma \ref{lemma_dom_grads}). In the stationary case, the OGA \rev{experts} with small learning rates obtain up to 12\% better throughput than the highest-learning-rate one. On the other hand, in the real case, the highest-learning-rate expert achieves 65\% more throughput. These findings support the claims that in more stationary scenarios, smaller-learning-rate \rev{experts} perform better \rev{as they update their decisions more carefully}, while in volatile/real cases, it is the opposite\rev{, as higher-learning-rate experts adapt more quickly to the rapid changes of the environment. 
}
\smallskip

\color{black}

\color{black}
\subsection{Regret Analysis}
\color{black}

Sec. \ref{sec:algorithm} shows that the regret guarantees hold for any benchmark; however, computing the benchmark that has full knowledge of the future is computationally intensive for mixed-integer programs \cite{kelleler-jsac23, Bragin2022}. To facilitate, therefore, the average dynamic regret calculation, the optimal \texttt{Oracle} considered solves the optimization problem \textit{in every step} using \textit{CVXPY} \cite{diamond2016cvxpy}, \rev{and a smaller number of $I=20$ users}. Apart from the \texttt{CONTRA}, we show the average dynamic regret of 3GPP threshold-based benchmarks. In the sequel, we refer to the 3GPP CHO benchmarks as $(CL, TTT)$, where the first element denotes the number of cells that are prepared in each slot, provided they are the highest-SINR cells for more than $TTT$ consecutive slots (similar to the A3 event); while for the 3GPP THO benchmarks, only $TTT$ is used. For instance, the benchmark $(1,2)$ prepares one cell for each user if it remains the highest-SINR for more than two slots.

\begin{figure}[t]
    \centering
    \subfigure[]{\label{fig:tilde_g_stationary}\includegraphics[scale=0.35]{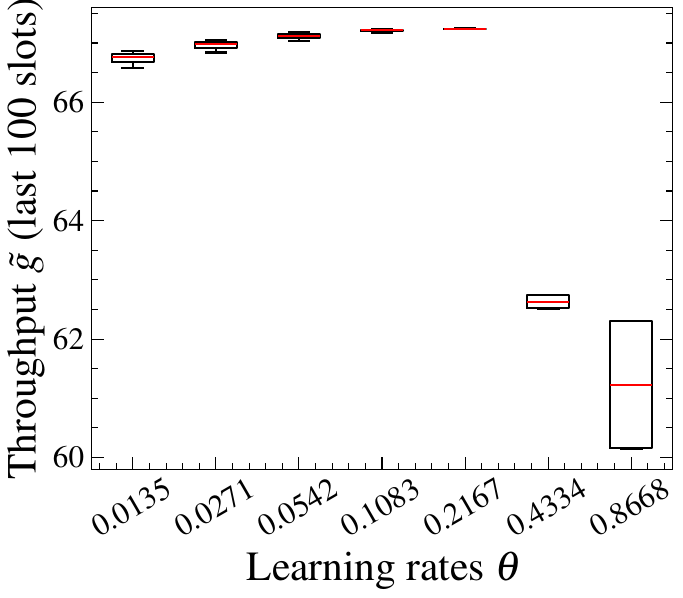}}
    \subfigure[]{\label{fig:tilde_g_real}\includegraphics[scale=0.35]{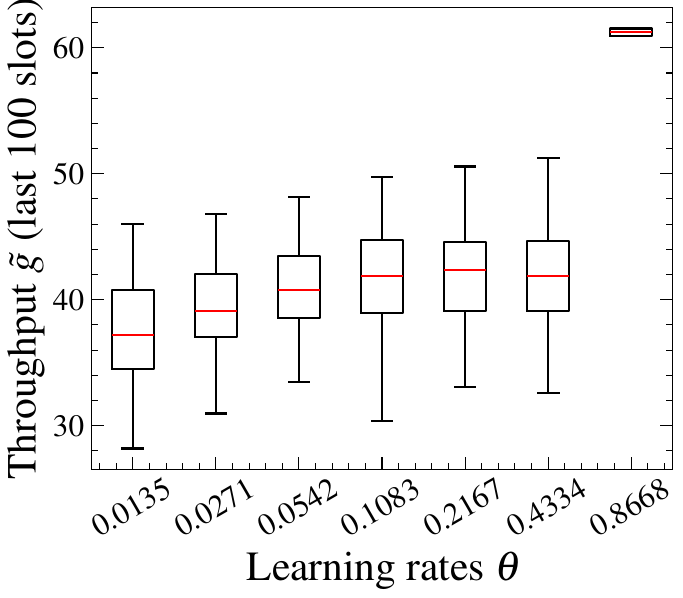}}
    \caption{Cumulative throughput $\tilde g$ for the last 100 of 1k slots for different learners in (a) stationary and (b) real SINR case.
    }
    \vspace{-6mm}
    \label{fig:tilde_g_real_stationary}
\end{figure}

Figure \ref{fig:avg_regret_volatile} shows the average dynamic regret of our proposed algorithm and the 3GPP-compliant competitors in the volatile case. In Figures \ref{fig:avg_regret_volatile_cont} and \ref{fig:avg_regret_volatile_discr}, we show the continuous, $\bm z \in \mathcal Z^{\text{c}}$, and discrete, $\bm z \in \mathcal Z$, decisions, respectively. In both cases, we observe that the average dynamic regret converges towards zero for $T\!=\!5$k slots; and in the latter case, where the decisions are actually implementable in practice, \texttt{CONTRA} outperforms the best-performing CHO-only 3GPP benchmark by 89.5\%. The average dynamic regret of the THO-only benchmarks stays almost constant for all slots (``stuck'' in sub-optimal decisions). Lastly, we observe a relatively small discretization error of 16.2\% measured between the continuous and discrete decision plots at $t=5$k.

From Figure \ref{fig:avg_regret_real}, the average dynamic regret converges again towards zero for the real SINR cases, showing that \texttt{CONTRA} is adaptable in all scenarios, with the gap between the continuous and discrete decisions being solely 1.83\%. On the other hand, the best CHO-only (THO-only) 3GPP-compliant benchmark attains 74\% (28.2\%) higher (lower) dynamic regret at $t\!=\!5$k. Yet, the THO-only benchmark is stuck in suboptimal (non-diminishing regret) decisions.

\begin{figure}[t]
    \centering
    \subfigure[]{\label{fig:avg_regret_volatile_cont}\includegraphics[scale=0.35]{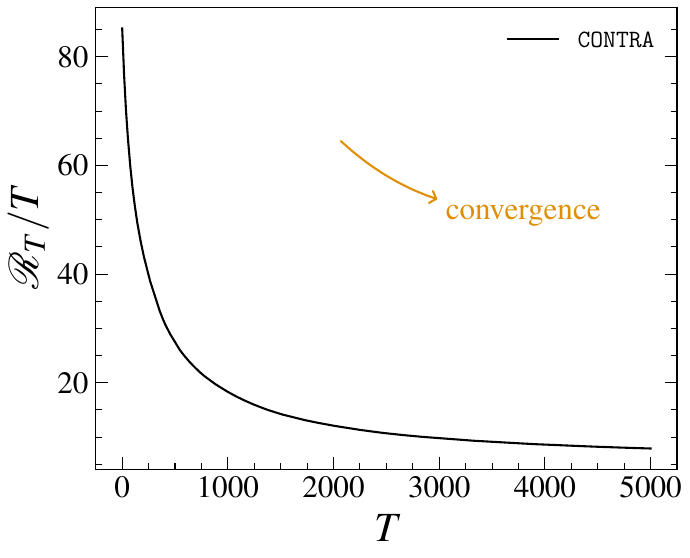}}\ \ \ 
    \subfigure[]{\label{fig:avg_regret_volatile_discr}\includegraphics[scale=0.35]{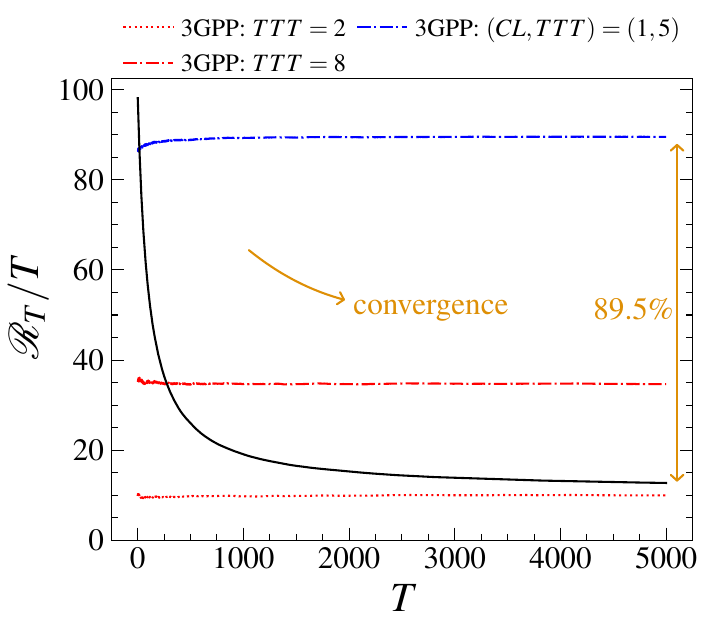}}
    \caption{Volatile scenario (SINR changes every 10 slots), for $T = 5$k slots, for the (a) continuous and (b) discrete decisions.
    }
    \vspace{-8mm}
    \label{fig:avg_regret_volatile}
\end{figure}

\begin{figure}[t]
    \centering
    \subfigure[]{\label{fig:avg_regret_real_cont}\includegraphics[scale=0.35]{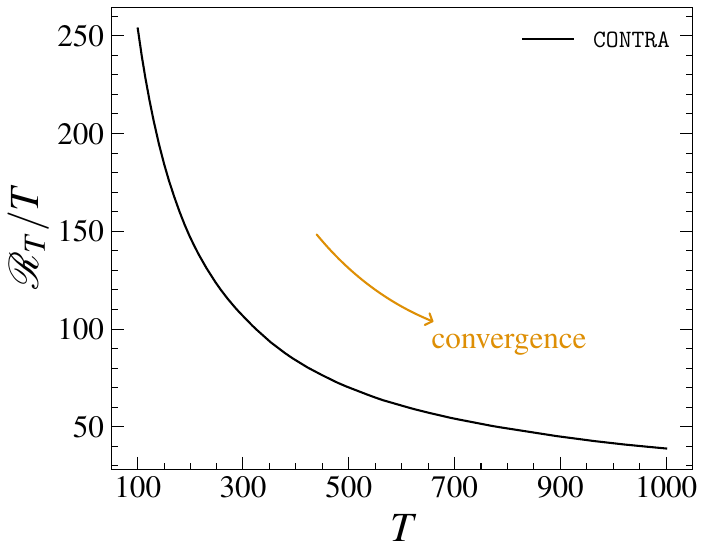}}\ \ \ 
    \subfigure[]{\label{fig:avg_regret_real_discr}\includegraphics[scale=0.35]{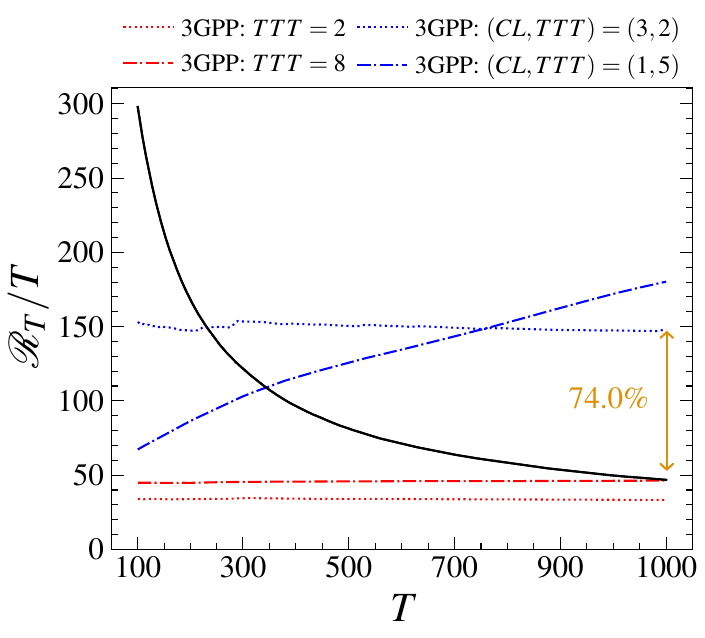}}
    \caption{Real-data SINR scenario for $T\!=\!1$k, for the (a) continuous and (b) discrete decisions.}    
    \label{fig:avg_regret_real}
\end{figure}

\color{black}

\rev{\subsection{Comparison with Reinforcement Learning}}

To compare \texttt{CONTRA} with RL approaches, we need to adjust the definitions of \emph{RL state} and \emph{RL action}. We define the \emph{RL state} with the SINRs of the previous time slots, see e.g., \cite {kelleler-jsac23, ppo_scc}. We intentionally avoid incorporating other UE features, such as velocity or ping-pong counters, to remain as aligned as possible with our problem formulation. We also define an \emph{RL action} that determines whether each user $i \in \mathcal I$ is THO- or CHO-enabled and how many cells to prepare in the latter case. In principle, 
the CHO decisiona should be modeled as the selection of any subset of candidate cells. This formulation captures the ``true’’ decision space of CHO, but in the case of RL, it leads to an exponential growth in the number of possible actions; specifically, $J + (2^J - 1)$ per user. The first term, $J$, corresponds to actions regarding THOs: a user can be assigned to any of the available $J$ cells. The second term (exponential), $2^J - 1$, corresponds to any non-empty subset of the $J$ cells that can get prepared (minus the empty set). 

Even when viewed purely from the search-space perspective, this exponential growth makes RL infeasible for moderate $J$, as exploration becomes prohibitively expensive. In contrast, \texttt{CONTRA} maintains a tractable search-space, avoiding the combinatorial explosion inherent to this RL formulation. To avoid the exponential explosion, we adopt a more compact parameterization for the RL approaches with only $2J$ actions per user. Here, the first $J$ terms still correspond to THOs, while the rest $J$ to CHOs, where the user prepares up to the top-$J$ strongest cells in terms of SINR in the previous slots. The reward/objective remains as in \texttt{CONTRA} (see Sec. \ref{sec:system_model}), namely, throughput minus switching costs.

The three RL benchmarks are: \textit{(i)} policy-gradient methods without a critic, through \texttt{REINFORCE} algorithm \cite{reinforce_williams}, \textit{(ii)} value-based methods through Deep Q-Network (\texttt{DQN}) \cite {kelleler-jsac23}, and \textit{(iii)} actor-critic methods with proximal updates through Proximal Policy Optimization (\texttt{PPO}) \cite{ppo_scc}. 

\texttt{REINFORCE} is a simpler method, that tries actions, observes total rewards, and adjusts the policy toward those that performed well. It does not estimate how good a state is (no critic), which makes it simple but noisy and data-hungry. More precisely, it implements a Monte Carlo policy-gradient method without a critic. Its policy network is a two-layer multilayer perceptron (MLP) with 128 neurons per layer and Tanh activations, followed by a softmax output layer producing a categorical distribution over all possible actions. After each episode, the algorithm computes discounted returns and updates the policy parameters using gradients of the log-probabilities weighted by returns. The optimizer is Adam with a learning rate of $3\times10^{-4}$.

\texttt{DQN} does not directly learn a policy; instead, it learns an action--value function $Q(state,action)$ that shows how good each action is. The agent then picks the action with the highest $Q$, allowing though a random exploration. More specifically, \texttt{DQN} consists of two hidden layers with 128 neurons and ReLU activations. Exploration follows an $\varepsilon$-greedy strategy with $\varepsilon_{\text{start}} = 1.0$, $\varepsilon_{\text{end}} = 0.05$, and exponential decay constant 200. A replay buffer of size 50k and minibatch size 64 are used, with soft target-network updates ($\tau = 0.01$) to improve stability. The loss is the mean-squared temporal-difference error, optimized with Adam at a learning rate of $10^{-3}$.

\texttt{PPO} combines both ideas. It learns an \emph{actor} (policy) and a \emph{critic} (value estimate), and ensures stable updates by clipping how much the policy can change each step. Both the actor and critic share a two-layer MLP backbone (128 neurons per layer, Tanh activations). The actor outputs a categorical distribution over all actions, while the critic outputs a scalar state-value estimate. Training uses generalized advantage estimation with $\lambda\!=\!0.95$ and the clipped objective parameter $\epsilon\!=\!0.2$; also, minibatches of 64 samples for 3 epochs per cycle are used. Both networks are optimized using Adam with rate $3\times10^{-4}$.

Unlike \texttt{CONTRA}, which learns and adapts online with a single exposure to the environment, these RL methods are trained over hundreds of simulated episodes, effectively granting them repeated access to the same system dynamics. This allows RL algorithms to asymptotically approximate the optimal policy under repeated trials, whereas \texttt{CONTRA} must adapt on-the-fly.

To showcase the disadvantages of RL even in a simple setting with four cells, we consider only two users with different types of mobility/SINR dynamics as can be seen in Figure \ref{fig:sinr_2UEs}: User 1 frequently switching between good and bad cells (high variability), and User 2 steadily moving toward a better SINR environment (gradual change). More precisely, the SINR of User 1 remains constant for 100 slots, then changes during a 50-slot transition period, and subsequently stabilizes again for another 100 slots, and so on. 

\begin{figure}[t]
    \centering
    \includegraphics[width=0.7\linewidth]{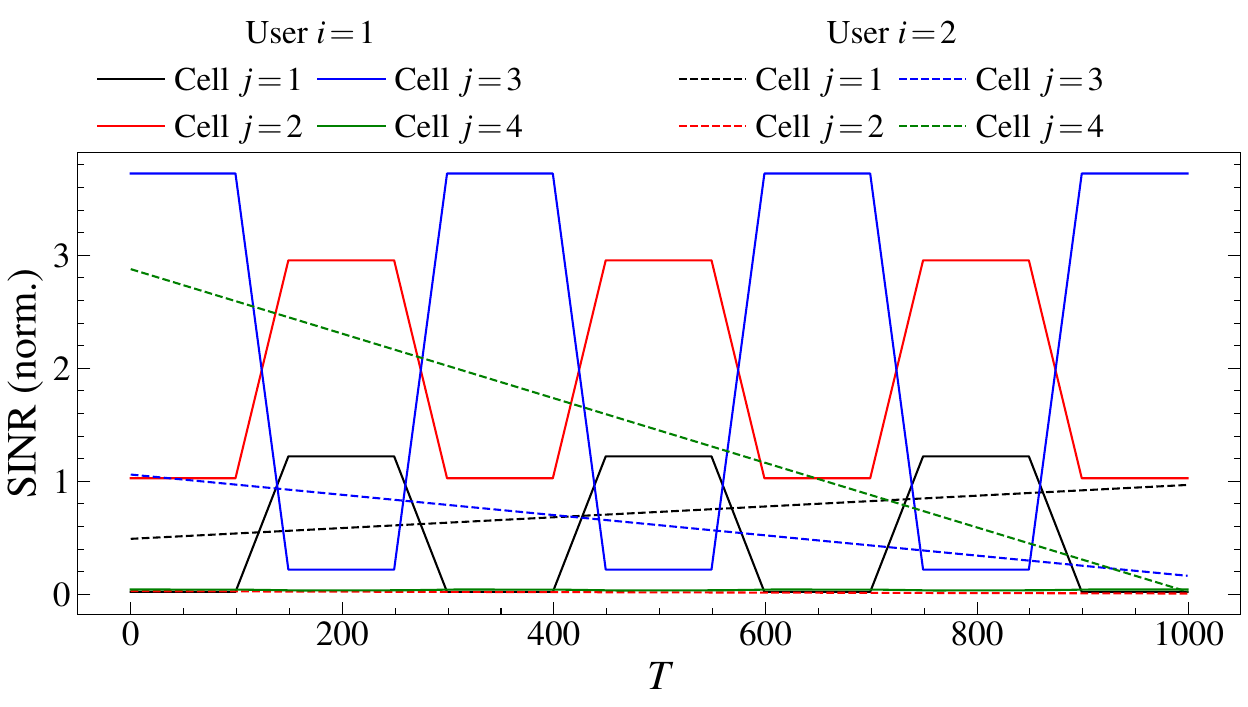}
    \caption{\rev{Simple deep-dive example for $I=2$ UEs, $J=4$ cells, and $T=1$k slots, with two SINR modes: one with more dynamic/abrupt changes (continuous lines) and one with more gradual changes (dotted lines).}}
    \vspace{-4mm}
    \label{fig:sinr_2UEs}
\end{figure}

In Figure \ref{fig:total_thr_costs_vs_RL}(a), we run each algorithm once (i.e., one episode), since our proposed method requires only a single pass to produce its decisions and learn. As expected, RL algorithms, and specifically \texttt{REINFORCE} and \texttt{PPO}, perform poorly (negative utility) as they require many interactions over the same conditions to learn. On the other hand, \texttt{DQN} can improve its decision in the same episode, achieving, however, 53.6\% lower utility than \texttt{CONTRA}. In Figure \ref{fig:total_thr_costs_vs_RL}(b), we allow the RL benchmarks to learn this simple scenario by running them for multiple episodes and reporting their final utilities. Compared to our algorithm that requires a single episode/pass to reach $\approx 2.5$k total utility, \texttt{REINFORCE}, \texttt{DQN}, and \texttt{PPO} reach until 95.2\%, 88.0\%, and 97.9\% of its performance, by running for 1.6k, 100, and 200 episodes, respectively. 

Finally, in the real SINR scenario, the performance gap becomes even more pronounced. From Figure \ref{fig:util_manyepisode_real}, and compared to our algorithm that requires a single pass to reach $\approx $175k total utility, \texttt{REINFORCE}, \texttt{DQN}, and \texttt{PPO} reach until 54.9\%, 52.6\%, and 76.6\% of its performance, by running for 2k, 500, and 500 episodes, respectively. The gap thus widens significantly in realistic network conditions, showing that \texttt{CONTRA} attains substantially higher performance without the extensive training required by RL-based approaches. Even with this simplified assumption, our approach outperforms the RL in the experiments executed with real SINR/UEs. This poor performance of RL approaches is expected, as they do not adapt to volatile or adversarial environments (known to converge under stationarity only), and do not offer performance guarantees as those we provide in Lemma \ref{lemma_main}.

\begin{figure}[t]
    \centering
    \subfigure[]{\includegraphics[width=0.42\linewidth]{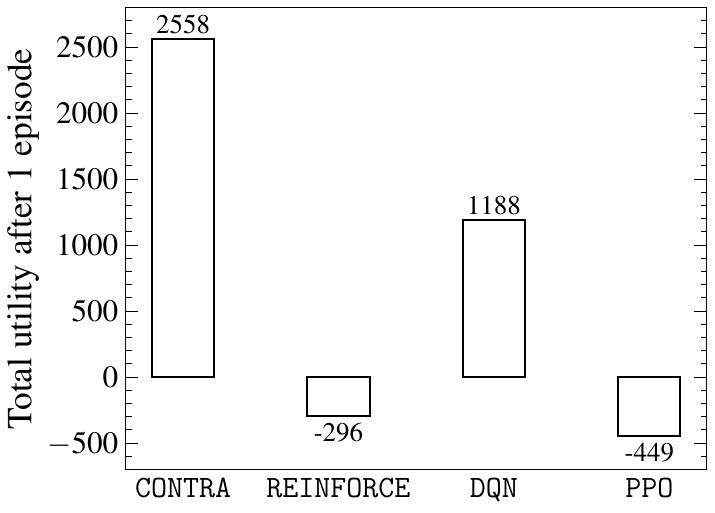}\label{fig:util_1episode_easy}}
    \subfigure[]{\includegraphics[width=0.55\linewidth]{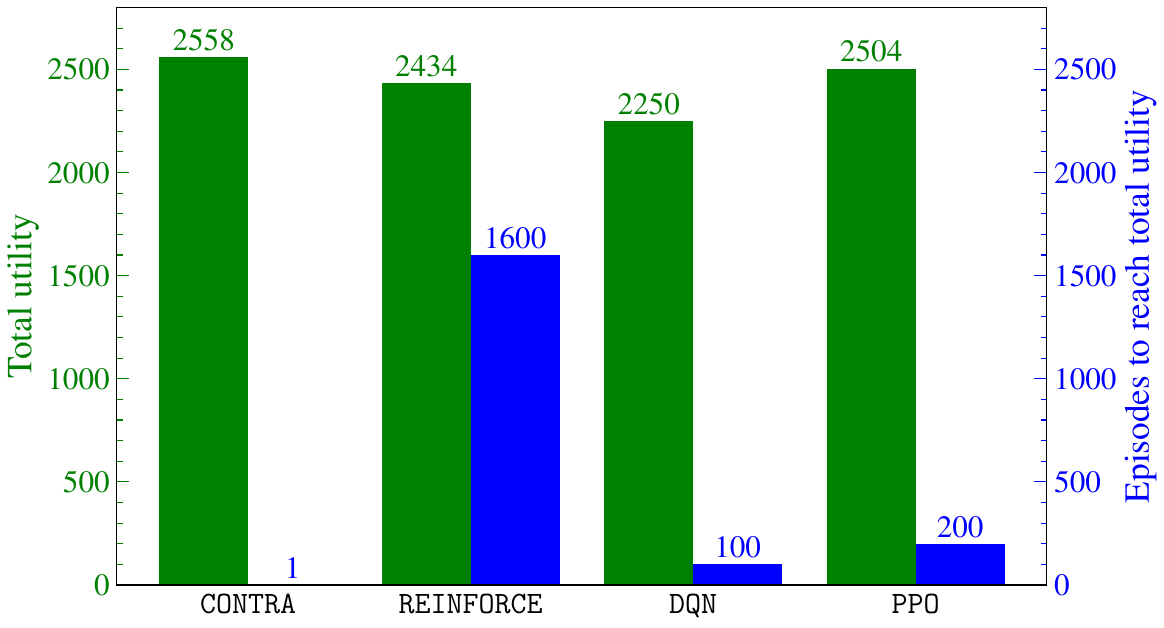}\label{fig:util_manyepisode_easy}}
    \caption{\rev{Total utility for (a) one episode and (b) multiple episodes, of \texttt{CONTRA} vs three RL benchmarks for a simple deep-dive example with $I=2$ UEs, $J=4$ cells, and $T=1$k slots, with the SINRs of Figure \ref{fig:sinr_2UEs}.}} 
    \label{fig:total_thr_costs_vs_RL}
\end{figure}

\begin{figure}[t]
    \centering
    \includegraphics[width=0.71\linewidth]{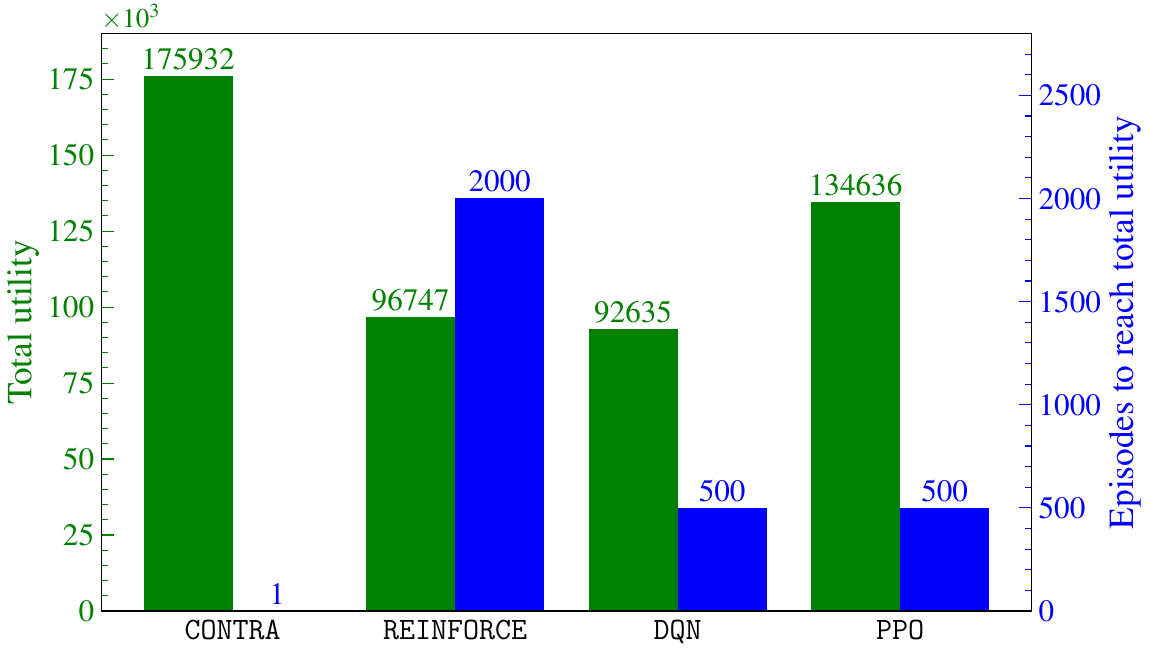}
    \caption{\rev{Total utilities (i.e., $\tilde{g}_t$) and number of episodes required to reach these values, for the real SINR scenario with $I=100$ UEs and $T=1$k slots.}}
    \label{fig:util_manyepisode_real}
    \vspace{-6mm}
\end{figure}

\color{black}

\section{Conclusions}\label{sec:conclusions}

\rev{This study is motivated by the demanding mobility support requirements in 6G, technological advances in handover (HO) mechanisms, and the availability of AI-based management solutions in O-RAN}. Its goal is to address a challenge that is critical as 5G deployments mature and 6G systems emerge: executing HOs rapidly and reliably, especially in dense deployments and high-frequency bands, where $(i)$ traditional HO (THO) mechanisms exhibit high failure rates and thus, increased delays, and $(ii)$ the newly introduced by 3GPP, Conditional HOs (CHOs), which tackle these issues by enabling proactive cell reservations, raise intricate trade-offs in signaling and resource utilization. To better understand these issues, we analyze countrywide datasets from a major operator in Europe; and the findings underscore the need for more adaptive and robust HO control. Motivated by this, we propose \texttt{CONTRA}, a meta-learning, provably-optimal, minimal-assumption, O-RAN-compatible algorithm that jointly optimizes THO and CHO and adapts to runtime observations, demonstrating significant improvements in real-world and synthetic scenarios. \rev{While our work focuses on O-RAN, the core ideas can be applied to other 3GPP-compliant RAN architectures that support real-time measurements and centralized decision making.}

\minorrev{\section*{Acknowledgments}
This work has been supported by the National Growth Fund through the Dutch 6G flagship project ``Future Network Services'', and the European Commission through Grant No. 101139270 (ORIGAMI), 101192462 (FLECON-6G), and 101167288 (ANTENNAE).}

\appendices
\ifCLASSOPTIONcaptionsoff
  \newpage
\fi

\bibliography{references_V2.bib}
\bibliographystyle{IEEEtran}

\appendix\label{sec:appendix_proofs}

\noindent\textbf{Proof of Lemma \ref{lemma_dom_grads}.}
First,
we observe that the diameter of the domain is bounded, i.e., $\| \bm z \!-\! \bm z'\|_2 \leq \! \| \bm x \!-\! \bm x'\|_2 \!+\! \| \bm y' \!-\! \bm y'\|_2 \! \leq \!
\sqrt{2I_{\text{THO}}} + \sqrt{I_{\text{CHO}}(J-1)} \!=\! D$, since the maximum distance for the former occurs when assignments $\bm x$ and $\bm x'$ are all different, and for the latter, when all cells are prepared in $\bm y$ and one cell is prepared per UE in $\bm y'$ (i.e., minimum preparation). Similarly, we prove the bound for the $C_t$-norm and its dual.

Finally, we compute the gradients component-wise and, for simplicity, we omit the time $t$ index. For $i \in \mathcal I_{\text{THO}}, j \in \mathcal J$ (similarly for $i \in \mathcal I_{\text{CHO}}$ and the $y_{ij}$), and since ${\partial l_j}/{\partial x_{ij}} = 1$, we have ${\partial \tilde g_t}/{\partial x_{ij}} = \log c_{ij} -  \log l_j - 1.$
Using $c_n \leq c_{\max}$ and $l_j \leq I_{\text{THO}} + I_{\text{CHO}}=I$, we can write:
$
\left| {\partial g_t}/{\partial x_{ij}} \right| \leq \max\left\{ \log c_{\max} - 1,\ \log I + 1 \right\} = M.
$
The proof is finalized, as there are $I \, J$ total components in the gradient vector; and likewise for the $C_t$-norm.

\noindent\textbf{Proof of Lemma \ref{lemma_main}.}
The proof tailors the result of \cite{kalntis_infocom25, zhang-smoothed-ol, zhang-18}, which, however, solve the UE-cell association problem (THO only). Eq. \eqref{regret-metric} can be rewritten: $\mathbb{E}\left[{\mathcal R}_T\right]=$
\begin{align}
&\sum_{t=1}^T \!\! \Big(\big(\tilde g_t(\bm z_t^{\star}) - \|\bm z_t^{\star}\!-\!\bm  z_{t-1}^{\star}\|_{C_t}\big) \! -\! \big(\tilde g_t(\bm{z}_t^\text{m})\!-\! \|\bm z_t^\text{m}\!-\!\bm  z_{t-1}^\text{m}\|_{C_t} \big)\!\Big) \notag \\ & \!\!+\! 
    \mathbb{E}\!\Bigg[\!\sum_{t=1}^T\!\! \Big(\!\!\big(\tilde g_t(\bm{z}_t^\text{m}) \!-\! \notag \|\bm z_t^\text{m}\!-\!\bm  z_{t-1}^\text{m}\|_{C_t}\big) \!-\! \big(\tilde g_t(\bm z_t) \!-\! \|\bm z_t\!-\!\bm z_{t-1}\|_{C_t}\big)\!\!\Big)\!\!\Bigg], \label{eq:prep_and_discr_err}
\end{align}
where we bound the expected dynamic regret of the relaxed preparation decisions $\{\bm z_t^\text{m}\}_t$ and, afterwards, the error of the discrete decisions $\{\bm z_t\}_t$ as can be seen from the first and second terms, respectively. From the first term in eq. \eqref{eq:lemma_res}, sublinear dynamic regret can be achieved for the relaxed decisions. It follows by bounding the regret of each expert w.r.t. the benchmark and then the regret of the meta-learner w.r.t. any expert. For the extra cost/error introduced due to the quantization routine $Q_{\mathcal{Z}}$, we bound the term as follows:
\[\mathbb{E}\!\!\left[ \!\sum_{t=1}^T\!\! \Big(\!\tilde g_t(\bm{z}_t^\text{m})\!-\! \|\bm z_t^\text{m}\!-\!\bm  z_{t-1}^\text{m}\|_{C_t} \!-\!\tilde g_t(\bm z_t) \!+\! \|\bm z_t\!-\!\bm  z_{t-1}\|_{C_t}\! \Big)\!\right] \!\!\stackrel{(*)}\leq \]
\[\Big( \!G+\sqrt{a_{\max} \!+\! b_{\max}}\Big) \! \sum_{t=1}^T \!\! \sqrt{\mathbb{E}\!\left[ \sum_{i=1}^I \! \sum_{j=1}^J \! \Big(\!\! {z}_{ij}(t) \!-\! \mathbb{E}[{z}_{ij}(t)] \Big)^{\!2} \! \right]} \!\!\! \stackrel{(**)}\leq\]
\[\hspace{-0.7cm}T\Big(G+\sqrt{a_{\max}+b_{\max}}\Big)\sqrt{I_{\text{CHO}}J/4 + I_{\text{THO}}(1-1/J)},
\]
\noindent leading to eq. \eqref{eq:lemma_res}. Specifically, for $(*)$, we use Jensen's inequality, the linearity of expectation, and Lipschitz continuity. For the latter, consider $\bm z_1, \bm z_2 \in \mathcal{Z}^\text{c}$, use the triangle inequality and Lipschitz constant $G$ of $g_t$ and the Euclidean norm, and define $\bm z_{t-1}\!=\!C$, to get:
\begin{align*}
    \Big|&\big(\tilde g_t(\bm{z}_1) -  \|\bm{z}_1 - C\|_{C_t} \big) - \big(\tilde g_t(\bm{z}_2) - \|\bm{z}_2 - C\|_{C_t} \big)\Big| \leq \\ &
    \big|\tilde g_t(\bm{z}_1) - \tilde g_t(\bm{z}_2)\big| + \big|\|\bm{z}_1 - C\|_{C_t} - \|\bm{z}_2 - C\|_{C_t}\big|  \leq \\ &
    \Big(G+ \sqrt{a_{\max}+b_{\max}}\Big) \|\bm{z}_1 - \bm{z}_2\|.
\end{align*}
\noindent Lastly, note that $(**)$ calculates the variance of $\bm z_t$. For $\bm y_t \in \{0,1\}^{I_{\text{CHO}} \cdot J}$, the maximum variance each component can obtain is $1/4$ (variance per element is $q(1-q)$ due to Bernoulli trial with prob $q$, and maximizes for $q=0.5$); hence, we can upper bound the part for CHO-enabled UEs expression with $I_{\text{CHO}}J/4$. For $\bm x_t \in \{0,1\}^{I_{\text{THO}} \cdot J}$, the binary vector is subject to a simplex
per user; thus, each sum of THO-enabled users w.r.t. $j$ is upper bounded by variance of $1-1/J$, and the overall part of the term in the square root has bound $I_{\text{THO}}(1-1/J)$.

\end{document}